\documentclass[journal]{IEEEtran}
\usepackage[english]{babel}
\usepackage[latin1]{inputenc}
\usepackage{enumerate}
\usepackage{color}
\usepackage[T1]{fontenc}
\usepackage{subfigure}
\usepackage{dsfont}
\usepackage[final]{graphicx}
\usepackage[T1]{fontenc}
\usepackage{amsmath}
\usepackage{mathtools}
\usepackage{amsthm}
\usepackage{amstext}
\usepackage{amssymb}
\usepackage{mathrsfs}
\usepackage{cite}
\usepackage{mathtools}
\usepackage{tikz}
\usetikzlibrary{arrows,positioning}

\newcommand{\appsec}{
\renewcommand{\thesubsection}{\Alph{subsection}}
}

\def\R{\mathbb{R}}
\def\cE{\mathbb{E}}

\def\cl{\mathop{\rm cl}}
\def\e{\mathop{\rm e}}
\def\N{\mathop{\rm N}}
\def\cdet{\mathop{\rm det}}

\def\B{{\mathcal B}}
\def\D{{\mathcal D}}

\def\H{{\mathcal H}}

\def\G{{\mathcal G}}
\def\C{{\mathcal C}}

\def\P{{\mathcal P}}

\def\M{{\mathcal M}}
\def\S{{\mathcal S}}
\def\L{{\mathcal L}}
\def\cR{{\mathcal R}}

\def\sPr{{\mathsf{Pr}}}
\def\sUnif{{\mathsf{Unif}}}
\def\sBer{{\mathsf{Bernoulli}}}

\def\sX{{\mathsf X}}
\def\sY{{\mathsf Y}}

\def\sU{{\mathsf U}}
\def\sV{{\mathsf V}}
\def\sW{{\mathsf W}}

\theoremstyle{remark}
\newtheorem{example}{Example}
\newtheorem{definition}{Definition}
\newtheorem{theorem}{Theorem}
\newtheorem{corollary}{Corollary}

\newtheorem{proposition}{Proposition}
\newtheorem{lemma}{Lemma}
\theoremstyle{remark}
\newtheorem{remark}{Remark}

\IEEEoverridecommandlockouts

\allowdisplaybreaks

\begin{document}
\sloppy
\title{Output Constrained Lossy Source Coding with Limited Common Randomness
\thanks{The authors are with the Department of Mathematics and Statistics,
     Queen's University, Kingston, ON, Canada,
     Email: \{nsaldi,linder,yuksel\}@mast.queensu.ca}
\thanks{This research was supported in part by
  the Natural Sciences and
  Engineering Research Council (NSERC) of Canada.}
\thanks{The material in this paper was presented in part at the
52nd Annual Allerton Conference on Communication, Control and Computing,
Monticello, Illinois, Oct.\ 2014.}
}
\author{Naci Saldi, Tam\'{a}s Linder, Serdar Y\"uksel}
\maketitle
\begin{abstract}
  This paper studies a Shannon-theoretic version of the generalized distribution
  preserving quantization problem where a stationary and memoryless source is
  encoded subject to a distortion constraint and the additional requirement that
  the reproduction also be stationary and memoryless with a given
  distribution. The encoder and decoder are stochastic and assumed to have
  access to independent common randomness.  Recent work has characterized the
  minimum achievable coding rate at a given distortion level when unlimited
  common randomness is available. Here we consider the general case where the
  available common randomness may be rate limited. Our main result completely
  characterizes the set of achievable coding and common randomness rate pairs at
  any distortion level, thereby providing the optimal tradeoff between these two
  rate quantities. We also consider two variations of this problem where we
  investigate the effect of relaxing the strict output distribution constraint
  and the role of `private randomness' used by the decoder on the rate
  region. Our results have strong connections with Cuff's recent work on
  distributed channel synthesis.  In particular, our achievability proof
  combines a coupling argument with the approach developed by Cuff, where
  instead of explicitly constructing the encoder-decoder pair, a joint
  distribution is constructed from which a desired encoder-decoder pair is
  established. We show however that for our problem, the separated solution of
  first finding an optimal channel and then synthesizing this channel results in
  a suboptimal rate region.

\end{abstract}


\section{Introduction}
\label{sec1}

In this paper, we aim to characterize the achievable
rate distortion region for the generalized distribution preserving randomized source coding problem,
where the rate region measures both the coding rate and the rate of common
randomness shared between the encoder and the decoder. To give a more precise definition of the problem, consider the communication system in Fig.~\ref{fig1}.

\begin{figure}[h]
\centering
\tikzstyle{int}=[draw, fill=white!20, minimum size=3em]
\tikzstyle{init} = [pin edge={to-,thin,black}]
\begin{tikzpicture}[node distance=4cm,auto,>=latex']
    \node [int] (a) {Encoder};
    \node (b) [left of=a,node distance=2cm, coordinate] {a};
    \node [int] (c) [right of=a] {Decoder};
    \node [coordinate] (end) [right of=c, node distance=2cm]{};
    \draw[<->, >=latex', shorten >=2pt, shorten <=2pt, bend left=45, thick, dashed]
    (a.north) to node[auto] {Rate $R_c$}(c.north);
    \path[->] (b) edge node {$X^n$} (a);
    \path[-]  (a) edge node {Rate $R$} (c);
    \path[->] (c) edge node {$Y^n$} (end) ;
\end{tikzpicture}
\caption{Randomized source coding with limited common randomness.}
\label{fig1}
\end{figure}
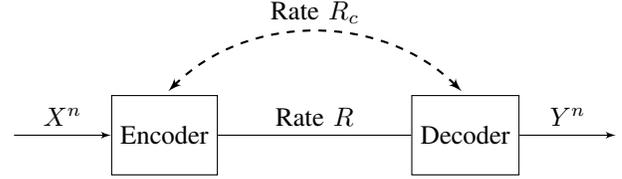

The source block $X^n=(X_1,\ldots,X_n)$ consists of $n$ independent drawings of a random variable
$X$ which takes values in a set $\sX$ and has distribution $\mu$.
The stochastic encoder takes the source and the common randomness, which is
available at rate $R_c$ bits per source symbol, as its inputs and produces an output
at a rate $R$ bits per source symbol. Observing the output of the encoder and
the common randomness, the decoder (stochastically) generates the output
(reconstruction) which takes values from a reproduction alphabet $\sY$. Here
$\sX=\sY$ is either a finite set or the real line.  The
common randomness is assumed to be independent of the source.  As usual, the
fidelity of the reconstruction is characterized by the expected distortion
\begin{align}
  \cE \biggl[\frac{1}{n} \sum_{i=1}^{n} \rho(X_i,Y_i)\biggr], \nonumber
\end{align}
where $\rho:\sX\times\sY\rightarrow [0,\infty)$ is a distortion
measure. However, unlike in the standard rate distortion problem, we require
that the output $Y^n=(Y_1,\ldots,Y_n)$ be a sequence of independent and
identically distributed (i.i.d.) random variables \emph{with a given common
distribution~$\psi$}.

For $D\ge 0$, a rate pair $(R,R_c)$ is said to be \emph{achievable} at distortion
level $D$ if, for any $\varepsilon>0$ and all $n$ large enough, there exists a
system as in Fig.~\ref{fig1} with coding rate $R$ and common randomness rate
$R_c$, such that the distortion of the system is less than $D+\varepsilon$ and
the output distribution constraint for $Y^n$ holds.  The main problem considered
in this paper is finding the set of all achievable rate pairs, denoted by
$\cR(D)$.

The communication system depicted in Fig.~\ref{fig1} is a generalized version
of a \emph{randomized quantizer (source code)} where the encoder and decoder is usually
assumed to have access to unlimited common randomization.
Randomized (dithered) uniform quantizers were originally introduced in signal processing by Roberts
\cite{Rob62}, where he observed that adding random noise to an image signal before
uniform quantization and subtracting the noise before reconstruction may result in perceptually more
pleasing images. Versions of dithered uniform quantizers were analyzed by Schuchman
\cite{Sch64} and Gray and Stockham \cite{GrSt93}. Under certain conditions,
dithering results in uniformly distributed quantization noise that is independent of the input
\cite{Sch64}, \cite{GrSt93}, which allows modeling the quantization process by an additive noise channel.
Related entropy-coded
dithered scalar and lattice quantizers have been extensively used in the
information theoretic literature to construct robust lossy compression schemes
with universal performance guarantees \cite{Ziv85,ZaFe92,ZaFe96b,Zam14}. Akyol and Rose
\cite{AkRo12-2}, \cite{AkRo13},  introduced a class of randomized \emph{nonuniform}
scalar quantizers obtained via applying companding to a dithered
uniform quantizer. Recently Li \emph{et al$.$}
\cite{LiKlKl10,LiKlKl11}  and Klejsa
\emph{et al$.$} \cite{KlZhLiKl13}  introduced and studied more
general classes of randomized quantizers that are
\emph{distribution-preserving}, i.e., the quantizer output is restricted to have
the same distribution as the source. The distribution-preserving property of
these quantizers is reported to significantly improve the perceptual quality of
the reconstruction in audio and video coding. Note that if in Fig.~\ref{fig1} we
set the distribution $\mu$ of the $X_i$ to be equal to the distribution $\psi$
of the $Y_i$, we obtain a distribution-preserving quantizer.

In our recent work \cite{SaLiYu13b,SaLiYu13} we studied a generalized version of
distribution-preserving randomized quantization where the output is constrained
to have a given distribution which may be different from the source
distribution. The main focus there was to develop an abstract and completely
general representation of finite-dimensional randomized quantization and to
study the existence and structural properties of optimal generalized
distribution preserving quantizers. Moreover, \cite{SaLiYu13} also
considered the asymptotic performance in the limit of infinite block length. In
particular, a rate distortion theorem was obtained for
stationary and memoryless sources under the assumption that the output must also
be a stationary and memoryless process and common randomness (in the form of a
random variable uniformly distributed on the unit interval $[0,1]$) is
shared by the encoder and the decoder. This situation corresponds to formally
setting
$R_c=\infty$ in Fig.~\ref{fig1}. In particular, \cite[Theorem 7]{SaLiYu13} showed
for both finite and continuous source and reproduction alphabets that the set of
achievable coding rates for unlimited common randomness $R_c=\infty$, denoted by $\cR(D,\infty)$, is
\begin{align}
\cR(D,\infty) = \{R \in \R: R \geq I(X;Y), \, P_{X,Y} \in \G(D)\},
\nonumber
\end{align}
where $\G(D)$ is the set of probability distributions $P_{X,Y}$ of $\sX \times
\sY$-valued random variables $(X,Y)$  defined as
\[
\G(D)\coloneqq \{P_{X,Y}: P_X=\mu, P_Y=\psi,\cE
[\rho(X,Y)] \leq D\}.
\]
Thus the minimum coding rate at distortion $D$ is
the so-called ``minimum mutual information with constrained output $\psi$''
\cite{ZaRo01}  given by
\begin{equation}
\label{eq_lower_mutual}
I(\mu\|\psi,D) \coloneqq \min\{ I(X,Y): P_{X,Y}\in \G(D)\}.
\end{equation}
If $\G(D)$ is empty, we let $I(\mu\|\psi,D) = \infty$.

In this paper, we generalize the above rate distortion result by studying the
optimal tradeoff between the coding rate $R$ and common randomness rate $R_c$
for the system in Fig.~\ref{fig1}. In particular, we find a single-letter
characterization of the entire achievable rate region $\cR(D)$ of pairs
$(R,R_c)$. Apart from the theoretical appeal of obtaining a computable
characterization of the rate region via information theoretic quantities, this
investigation is also motivated by the fact that the common randomness rate
$R_c$ has a direct affect on the complexity of the system since each possible
value of the common randomization picks a different (stochastic) encoder and
decoder pair from a finite set whose size is proportional to $2^{nR_c}$.
We also consider two variations of the problem, in which we
investigate the effect of relaxing the strict output distribution
constraint and the role of private randomness used by the decoder on the rate region.
For both of these problems, we give the complete characterizations of the achievable rate pairs.

It is important to point out that the block diagram in Fig.~\ref{fig1} depicting
the generalized distribution preserving quantization problem has the same
structure as the system studied by Cuff \cite{Cuf08,Cuf13} to synthesize
memoryless channels up to vanishing total variation error.  Although many other
problems in information theory share a similar representation, the connection
with Cuff's work is more than formal.  The distortion and output distribution
constraints in our problem replaces the requirement in \cite{Cuf13} that the
joint distribution of the input $X^n$ and output $Y^n$ should arbitrarily well
approximate (in total variation) the joint distribution obtained by feeding the
input $X^n$ to a given memoryless channel. Using the main result
\cite[Theorem~II.1]{Cuf13} one can obtain  an inner
bound, albeit a loose one, for our problem. A good part of our proof consists of
tailoring Cuff's arguments in \cite{Cuf13} to our setup to obtain a tight
achievable rate region. Because of this, we will be adopting many of the
notations used in \cite{Cuf13}. We also note that unlike in the distributed
channel synthesis problem in \cite{Cuf13}, our results also allow for continuous
source and reproduction alphabets.

The rest of the paper is organized as follows. In Section~\ref{sec2} we
formalize the problem and present the main result giving the rate region
$\cR(D)$. Section~\ref{sec_con} discusses connections with Cuff's work on distributed channel synthesis. In Section~\ref{sec_spec} we investigate the extreme points of the
rate region at $R_c=0$ and $R_c=\infty$. In Section~\ref{sec exs} we present computable inner bounds for
double symmetric binary source and reproduction distributions under
the Hamming distortion, and for Gaussian source and reproduction distributions under the squared error distortion. In Section~\ref{sec3} two variations of the original problem are formulated and the associated achievable rate regions are described. The proof of the main result is given in Section~\ref{sec_proof}.

\subsection{Notation and Assumptions}
\label{sub1sec1}

In this paper, $\sX$ denotes the input alphabet and $\sY$ is the reconstruction
(output) alphabet such that $\sX=\sY$ is a finite set or $\sX=\sY=\R$. We assume
a distortion measure $\rho(x,y)= d(x,y)^p$, where $d$ is the metric on
$\sX$. Here, $p>0$ when $\sX$ is finite and $p=2$ when $\sX=\R$, in which case
we also assume that $d(x,y)=|x-y|$ (so that $\rho$ is the squared error) and
that the source distribution $\mu$ and the desired output distribution $\psi$
have finite second moments.  We note that we impose these
  restrictions on the distortion measure because in a key step of the
  achievability proof we need to invoke the triangle inequality.   For the
finite alphabet case, we let $\rho_{\max} \coloneqq \max_{x,y} \rho(x,y)$. For
any positive real number $R$, we define $[2^{nR}] \coloneqq \{1,\ldots,\lceil
2^{nR}\rceil\}$, where $\lceil 2^{nR}\rceil$ is the smallest integer greater
than or equal to $2^{nR}$.  $\sV^n$ will denote the $n$-fold Cartesian product
of a set $\sV$, the elements of which are $v^n=(v_1,\ldots,v_n)$, $v_i\in \sV$,
$i=1,\ldots,n$. A similar convention also applies to a sequence of random
variables which will be denoted by upper case letters. For any triple $(X,Y,U)$
of random variables or vectors, the notation $X-U-Y$ means that they form a
Markov chain in this order. For any random vector $U^n$, the random measure
$p_{U^n}$ denotes the empirical distribution of $U^n$.
The notation $V\sim \nu$ means that random variable $V$ has
distribution $\nu$. For any probability distribution $\nu$ on $\sV$, $\nu^n$
denotes the $n$-fold product distribution
$\underbrace{\nu\times\cdots\times\nu}_{\text{$n$-times}}$ on $\sV^n$.

\section{Problem Statement and Main Result}
\label{sec2}

Let $\{X_n\}_{n\geq1}$ be a
stationary and memoryless source (sequence of i.i.d.\ random variables) with
common distribution $\mu$ on source alphabet $\sX$, and let $K$ be a random variable
uniformly distributed over $[2^{nR_c}]$ which is independent of $X^n$. Here $K$
represents the common randomness that is shared between the encoder and the
decoder.

For a positive integer $n$ and nonnegative numbers $R$ and $R_c$, a $(n,R,R_c)$
\emph{randomized source code} is defined by an encoder $E=E_{J|X^n,K}$ and the
decoder $F_{Y^n|J,K}$, where $E$ is a regular conditional probability (see
\cite{Dud89}) on $[2^{nR}]$ given $\sX^n\times[2^{nR_c}]$ and $F$ is a regular
conditional probability on $\sY^n$ given $[2^{nR}]\times[2^{nR_c}]$. Hence,
letting $J$ and $Y^n$ be the output of the encoder and the decoder,
respectively, the joint distribution of $(K,X^n,J,Y^n)$ is given, in a somewhat
informal notation, by
\begin{align}
(K,X^n,J,Y^n) \sim    F_{Y^n|J,K} E_{J|X^n,K} P_K P_{X^n} \label{eq5}.
\end{align}
The distortion of the code is $ \cE[\rho_n(X^n,Y^n)] $, where
$\rho_n(x^n,y^n) \coloneqq \frac{1}{n}\sum_{i=1}^{n} \rho(x_i,y_i)$.

\begin{definition}
\label{def1}
For any nonnegative real number $D$ and desired output distribution $\psi$, the
pair $(R,R_c)$ is said to be \emph{$\psi$-achievable} if, for any
$\varepsilon>0$ and all sufficiently large $n$, there exists a randomized
$(n,R,R_c)$ source code such that
\begin{align}
\cE[\rho_n(X^n,Y^n)] &\leq D + \varepsilon \nonumber \\
Y^n &\sim \psi^n. \nonumber
\end{align}
\end{definition}
In the rest of this paper $\psi$ will be kept fixed, so we drop referring to
$\psi$ and simply write that $(R,R_c)$ is \emph{achievable}. For $D\ge 0$ we let
$\cR(D)$ denote the set of all achievable $(R,R_c)$ pairs.  The following
theorem, which is the main result in this paper, characterizes the closure of
this region in terms of an auxiliary random variable $U$ on alphabet $\sU$.

\begin{theorem}
\label{thm2}
For any $D\geq0$ the closure  $\cl\cR(D)$  of  $\cR(D)$ is given by
\begin{align}
&\hspace{-10pt}\cl\cR(D) = \L(D)  \nonumber \\
&\phantom{xx}\coloneqq \left\{\begin{array}{rcl}
(R,R_c) \in \R^2 &: &\exists P_{X,Y,U}\in\M(D) \text{ s.t. } \\
 R &\geq &I(X;U),  \\
 R+R_c &\geq &I(Y;U)
\end{array}\right\} \label{neq10},
\end{align}
where, for $\sX=\sY$ finite,
\begin{align}
\M(D) &\coloneqq \left\{\begin{array}{cl}
P_{X,Y,U} :& P_{X} = \mu, P_{Y} = \psi, \\
 &\cE[\rho(X,Y)]\leq D,\, X-U-Y, \\
 &|\sU|\leq |\sX|+|\sY|+1
\end{array}\right\}. \label{neq11}
\end{align}
When $\sX=\sY=\R$, the cardinality bound for $\sU$
in \eqref{neq11} is replaced by $\sU=\R$.
\end{theorem}

\subsection{Connections with Distributed Channel Synthesis}
\label{sec_con}

As mentioned before, Cuff's work on distributed channel synthesis \cite{Cuf13} is
intrinsically related to our problem. The main objective of \cite{Cuf13} is to
simulate a memoryless channel by a system as in Fig.~\ref{fig1}. To be more
precise, let $Q(y|x)$ denote a given  discrete memoryless channel with
input alphabet $\sX$ and output alphabet $\sY$ to  be simulated
(synthesized) for input $X$ having distribution $\mu$. Let $\pi=\mu Q$ be the
joint distribution of the resulting input-output pair $(X,Y)$.

\begin{definition}[\cite{Cuf13}]
\label{def4}
The pair $(R,R_c)$ is said to be \emph{achievable} for synthesizing a memoryless
channel $Q$ with input distribution $\mu$ if there exists a sequence of $(n,R,R_c)$
randomized source codes  such that
\begin{align}
\lim_{n\rightarrow\infty} \|P_{X^n,Y^n} - \pi^n\|_{TV} = 0, \label{eq_tvconv}
\end{align}
where $X^n \sim \mu^n$ is the memoryless source, $Y^n$ is the output of the
decoder, $\pi^n$ is the $n$-fold product of $\pi=\mu Q= P_X Q$, and $\|\,\cdot\,\|_{TV}$ is the total variation distance for
probability measures: $\|\gamma-\nu\|_{TV} \coloneqq \frac{1}{2} \sum_{v} |\gamma(v) - \nu(v)|$.
\end{definition}

\begin{theorem}{\cite[Theorem II.1]{Cuf13}}
\label{thm5}
The  closure $\C$  of the
set of all achievable $(R,R_c)$ pairs is given by
\begin{align}
\C= \S \coloneqq& \left\{ \begin{array}{rcl} (R,R_c)\in\R^2 &:& \exists P_{X,Y,U}\in\D \text{ s.t.} \\
 R&\geq &I(X;U), \\
 R+R_c&\geq &I(X,Y;U) \end{array} \right\} \label{eq30},
\end{align}
where
\begin{align}
  \D \coloneqq& \{P_{X,Y,U}: P_{X,Y} = \pi, X-U-Y, |\sU|\leq
  |\sX\|\sY|+1\}. \nonumber
\end{align}
Moreover, the total variation error goes to zero exponentially fast with respect to $n$ in the interior of $\mathcal{C}$.
\end{theorem}

This result can be used to obtain an achievable rate region (inner bound) for
our problem as follows: Let $\pi= P_{X,Y}$ be such that $P_X=\mu$, $P_Y=\psi$,
and $\cE[\rho(X,Y]\le D$. Applying Theorem~\ref{thm5} with this input
distribution and the channel induced by $P_{X,Y}$, consider an achievable rate
pair $(R,R_c)$ in \eqref{eq30}. Using basic results from optimal transport
theory  \cite{Vil09}  one can show that \eqref{eq_tvconv} and the fact that $\cE[\rho(X,Y)]\leq D$ imply the existence of a
sequence of channels, to be used at the decoder side, that when fed with $Y^n$,
produces output $\hat{Y}^n$  which has the exact distribution $\psi^n$ and which
additionally
satisfies
\[
 \limsup_{n\to \infty} \cE[\rho_n(X^n,\hat{Y}^n)]\le D.
\]
Augmenting the channel synthesis code with these channels at the decoder side
thus produces a sequence of valid codes for our problem, implying   that
the rate pair $(R,R_c)$ is achievable by our Definition~\ref{def1}.

Using the above argument, one can easily show that Cuff's result directly
implies (without resorting to Theorem~\ref{thm2}) the following
inner bound for $\cR(D)$. The proof is given in Appendix~\ref{sec6sub2}.

\begin{corollary}
\label{cor1}
For any $D\geq0$,
\begin{align}
&\hspace{-10pt}\cl\cR(D) \supset \S(D) \\
&\phantom{xx} \coloneqq \left\{ \begin{array}{rcl} (R,R_c)\in\R^2&:& \exists P_{X,Y,U}\in\H(D) \text{ s.t.} \\
      R&\geq &I(X;U), \\
      R+R_c&\geq &I(X,Y;U) \end{array} \right\} \label{eq10},
\end{align}
where
\begin{align}
  \H(D) \coloneqq& \left\{ \begin{array}{cl} P_{X,Y,U}:& P_{X} = \mu, P_{Y} = \psi, \\
      &\cE[\rho(X,Y)]\leq D, X-U-Y, \\
      &|\sU|\leq |\sX\|\sY|+1 \end{array} \right\}. \label{eq11}
\end{align}
\end{corollary}

In general, this inner bound is loose. For example, for $R_c=0$, only the
constraint $R\geq I(X,Y;U)$ is active in (\ref{eq10}) since $I(X,Y;U) \geq
I(X;U)$ always holds. Hence, letting $\S(D,0)$ denote the set of $R$s such that
$(R,0) \in \S(D)$, we obtain
\begin{align}
\S(D,0) = \{R \in \R: \exists P_{X,Y,U}\in\H(D)
\text{ s.t. } R\geq I(X,Y;U)\}. \nonumber
\end{align}
The minimum of $\S(D,0)$ can be written as
\begin{align}
&\min\{R \in
\S(D,0)\} \nonumber \\
&\phantom{xxx}=\min\{C(X;Y): P_{X,Y} \in \G(D)\}=:C_0(\mu\|\psi,D), \nonumber
\end{align}
where $C(X;Y)$ is Wyner's common information~\cite{Wyn75} defined for a given joint
distribution $P_{X,Y}$ by
\begin{align}
C(X;Y) \coloneqq \inf_{U: X-U-Y} I(X,Y;U) ,\label{neq6}
\end{align}
where the infimum is taken over all joint distributions $P_{X,Y,U}$ such that
$U$ has a finite alphabet and $X-U-Y$. However, the resulting rate $C_0(\mu\|\psi,D)$
is not optimal as Example~\ref{exm1} in Section~\ref{sec2sub1} will show.

The suboptimality of $C_0(\mu\|\psi,D)$ implies that a 'separated' solution
which first finds an 'optimal' channel and then synthesizes this channel is not
optimal for the constrained rate distortion problem we consider.

\section{Special Cases}
\label{sec_spec}

The extreme points at $R_c=\infty$ and $R_c=0$ of the rate region $\L(D)$ in our Theorem~\ref{thm2} are of
particular interest. Let $\L(D,R_c)$ be the set of coding rates $R$ such that
$(R,R_c) \in \L(D)$.

\subsection{Unlimited Common Randomness}
If $R_c=\infty$, then the effective constraint in
(\ref{neq10}) is $R \geq I(X;U)$. This was the situation originally
studied in \cite{SaLiYu13} where it was assumed that the common randomness is of
the form of a real-valued random variable that is uniformly distributed on the
interval $[0,1]$. Since $I(X;U) \geq I(X;Y)$ by the data processing
inequality and the condition $X-Y-Y$, we can set $U=Y$ to obtain $\min \{R \in \L(D,\infty)\} =
I(\mu\|\psi,D)$, recovering \eqref{eq_lower_mutual} and thus \cite[Theorem
7]{SaLiYu13}. Furthermore, for the finite alphabet case whenever
$R_c \geq H(Y|X)$, we have from (\ref{neq10}) that $R+R_c \geq I(X;U)+H(Y|X)\geq I(X;Y)+H(Y|X) = H(Y) \geq I(Y;U)$, so the
effective constraint is again $R\geq I(X;U)$. Considering $(X,Y)$ such that $P_{X,Y}$ achieves the minimum in
\eqref{eq_lower_mutual} and letting $U=Y$, we have
\begin{align}
\min \{R \in \L(D,R_c)\} &= I(\mu\|\psi,D) \label{neq3}
\end{align}
or equivalently
\begin{align}
\L(D,R_c) &= \L(D,\infty). \label{neq2}
\end{align}
Hence, $H(Y|X)$ is a sufficient common
randomness rate above which the minimum communication rate does not
decrease. In fact, letting
\[
  R_c^{\,\text{min}} = \min\{  R_c: \L(D,R_c) = \L(D,\infty)\}
\]
we can determine $  R_c^{\,\text{min}}$ in terms of the
so-called  \emph{necessary conditional
  entropy} \cite{Cuf13}, defined for a joint distribution $P_{X,Y}$ as
\begin{align}
H(Y \dag X ) \coloneqq \min_{f: X-f(Y)-Y} H(f(Y)|X) \nonumber
\end{align}
where minimum is taken over all functions $f:\sY\to \sY$ such that $
X-f(Y)-Y$. Using the discussion in \cite[Section VII-C]{Cuf10} one can verify
that $ R_c^{\,\text{min}}$ is the minimum of $H(Y \dag X )$ over all joint
distributions of $(X,Y)$ achieving the minimum in \eqref{eq_lower_mutual}.
Indeed, for any joint distribution $P_{X,Y}$ achieving the minimum in
\eqref{eq_lower_mutual}, any function $f$ with the property
\begin{align}
f(y) = f(\tilde{y}) \Leftrightarrow P_{X|Y}(\,\cdot\,|y) = P_{X|Y}(\,\cdot\,|\tilde{y}) \label{nneq1}
\end{align}
minimizes $H(f(Y)|X)$ and satisfies $X-f(Y)-Y$; that is, $H(Y \dag X ) = H(f(Y)|X)$.

In general, for an arbitrary output distribution $\psi$, it may not be true
  that $H(Y \dag X) = H(Y|X)$ for a joint distribution achieving the minimum
  in \eqref{eq_lower_mutual}. Therefore, the Markov chain $X-Y-Y$ does not
  necessarily achieve $R_c^{\min}$. However, in the special case where the
  rate-distortion  function
\[
  R(D)=  \min_{\psi} I(\mu\|\psi,D),
\]
is achieved by a unique output distribution $\psi$, we have the following
proposition.

\begin{proposition}
\label{rdprop}
  Assume the rate-distortion function $R(D)$ is achieved by the unique
  output distribution $\psi$. Then $H(Y
  \dag X) = H(Y|X)$ and the Markov chain $X - Y - Y$ (i.e., $U=Y$) achieves
  $R_c^{\min}$, where $(X,Y)$  achieve the rate-distortion function. In
  this case,  $R+R_c \ge H(Y)$ when $R=R(D)$.
\end{proposition}

\begin{proof}
  The proof is by contradiction.  Suppose that $H(Y \dag X) < H(Y|X)$. This
  implies the existence of a function $f$ with the property \eqref{nneq1} and
  $H(f(Y)|X) < H(Y|X)$. In particular, there exist $\bar{y},\tilde{y} \in \sY$
  such that $\bar{y} \neq \tilde{y}$, $P_Y(\bar{y}), P_Y(\tilde{y}) > 0$, and
  $P_{X|Y}(\,\cdot\,|\bar{y}) = P_{X|Y}(\,\cdot\,|\tilde{y})$. Without loss of
  generality we can assume
  $\cE[\rho(X,Y)|Y=\bar{y}] \leq \cE[\rho(X,Y)|Y=\tilde{y}]$. Define a new pair
  $(\tilde{X},\tilde{Y})$ with the joint distribution given by
  $P_{\tilde{X}|\tilde{Y}} = P_{X|Y}$ and $P_{\tilde{Y}}(y) = P_{Y}(y)$ if
  $y \in \sY\setminus\{\bar{y},\tilde{y}\}$ and
  $P_{\tilde{Y}}(\bar{y}) = P_{Y}(\bar{y}) + P_{Y}(\tilde{y})$ (so,
  $P_{\tilde{Y}}(\tilde{y}) = 0$). Hence,
  $\cE[\rho(\tilde{X},\tilde{Y})] \leq \cE[\rho(X,Y)]$. Since
  $P_{X|Y}(\,\cdot\,|\bar{y}) = P_{X|Y}(\,\cdot\,|\tilde{y})$, we have
  $H(\tilde{X}|\tilde{Y}) = H(X|Y)$ and $P_{\tilde{X}} = P_{X}=\mu$. Therefore,
  $I(\tilde{X};\tilde{Y}) = I(X;Y)$, and $(\tilde{X},\tilde{Y})$ also achieves
  the rate distortion function. But, $P_{\tilde{Y}} \neq P_Y=\psi$, which is a
  contradiction.
\end{proof}

\subsection{No Common Randomness} \label{sec2sub1}
Setting $R_c=0$ means that no common randomness is available.\footnote{Ram Zamir's
  question regarding the minimum coding  rate in this special case has
  inspired  our investigation of the general rate region $\cR(D)$.}  In this
case (\ref{neq10}) gives $R \geq \max\bigl(I(X;U),I(Y;U)\bigr)$.
Hence the minimum communication rate at distortion $D$ is given by
\[
\min \{R \in \L(D,0)\} = I_0(\mu\|\psi,D),
\]
where
\begin{align}
&I_0(\mu\|\psi,D) \nonumber \\* &\coloneqq \min \bigl\{
\max\bigl(I(X;U),I(Y;U)\bigr): P_{X,Y,U} \in \M(D) \bigr\}. \label{neq4}
\end{align}

Note that the minimum achievable coding rate $I_0(\mu\|\psi,D)$ is
\emph{symmetric} with respect to $\mu$ and $\psi$, i.e.,
$I_0(\mu\|\psi,D)=I_0(\psi\|\mu,D)$. This is clear from the definition
\eqref{neq4}, but can also be deduced from the operational meaning of
$I_0(\mu\|\psi,D)$ since in the absence of the common randomness $K$,
the encoder-decoder structure is fully reversible.  In
general such symmetry no longer holds for $\min\{R \in \cR(D,R_c)\}$ when
$R_c>0$.

The following lemma states
that $I_0(\mu\|\psi,D)$ is convex in $D$.  The proof simply follows from a
time-sharing argument and the operational meaning of $ I_0(\mu\|\psi,D)$
implied by Theorem~\ref{thm2}. It is given in Appendix~\ref{sec6sub1}.

\begin{lemma}
\label{lemma1}
$I_0(\mu\|\psi,D)$  is a convex function of $D$.
\end{lemma}

An upper bound for $I_0(\mu\|\psi,D)$ can be given in terms of Wyner's common
information. Since $\max \bigl(I(X;U),I(Y;U)\bigr) \leq I(X,Y;U)$,
we have $I_0(\mu\|\psi,D) \leq \min \{I(X,Y;U): P_{X,Y,U} \in \M(D)\}$. The
latter expression can also be written as
\begin{align}
\min \{C(X;Y): P_{X,Y} \in \G(D)\} =: C_0(\mu\|\psi,D). \label{neq5}
\end{align}
However, the resulting upper bound $
I_0(\mu\|\psi,D) \le C_0(\mu\|\psi,D)$ is not tight in general as the next
example shows.

\begin{example}\label{exm1}

  Let $\sX = \sY = \{0,1\}$, and let $\mu = \psi = \sBer(1/2)$, i.e., $\mu(0) =
  \mu(1) = \frac{1}{2}$.  Assume the distortion measure $\rho$ is the Hamming
  distance $\rho(x,y) = 1_{\{x \neq y\}}$  (which satisfies
  the assumptions in Section~\ref{sub1sec1}). If $X \sim \mu$ and $Y \sim \psi$,
  then the channel $P_{Y|X}$ from $X$ to $Y$ must be Binary Symmetric Channel
  (BSC) with some crossover probability $a_0$, i.e.,
\begin{align}
P_{Y|X}(\,\cdot\,|0) = 1-P_{Y|X}(\,\cdot\,|1) = \sBer(a_0). \nonumber
\end{align}
Wyner in \cite[Section 3]{Wyn75} showed that when $a_0 \in [0,1/2]$,
\begin{align}
C(X;Y) = 1 + h(a_0) - 2h(a_1), \nonumber
\end{align}
where $a_1 = \frac{1}{2}(1 - \sqrt{1-2a_0})$, and $h(\lambda) =
-\lambda\log(\lambda) - (1-\lambda)\log(1-\lambda)$. Define $C(a_0) \coloneqq 1
+ h(a_0) - 2h(a_1)$ which is decreasing and strictly concave in
$[0,1/2]$. Notice that $\cE[\rho(X,Y)] = a_0$ when
$P_{Y|X}=\mathsf{BSC}(a_0)$. Hence, for any $D \in [0,1/2]$, we have
\begin{align}
&C_0(\mu\|\psi,D)  \nonumber \\*
&= \min\{C(X;Y): P_{X,Y} \in \G(D)\} \nonumber \\
              &= \min\{C(X;Y): P_X=\mu, P_{Y|X}=\mathsf{BSC}(a_0), a_0 \leq D\}  \nonumber \\
              &= \min_{a_0 \leq D} C(a_0) = C(D) \nonumber
\end{align}
implying that $C_0(\mu\|\psi,D)$ is strictly concave for $D\in [0,1/2]$.
It is straightforward to prove that $C_0(\mu\|\psi,0)= I_0(\mu\|\psi,0)=1$ and $C_0(\mu\|\psi,1/2)=
I_0(\mu\|\psi,1/2)=0$. Therefore, by Lemma~\ref{lemma1} we have
\[
I_0(\mu\|\psi,D) < C_0(\mu\|\psi,D), \quad D\in (0,1/2).
\]
\end{example}

\medskip

\section{Examples}\label{sec exs}

In general determining the entire rate region $\L(D)$ in Theorem~\ref{thm2}
seems to be difficult even for simple cases. In this section we obtain
possibly suboptimal achievable rate regions (inner bounds) for two setups
by restricting the channels $P_{U|X}$ and $P_{Y|U}$ so that the resulting
optimization problem becomes manageable.

\subsection{Doubly Symmetric Binary Source}
\label{ex1}

In this section we obtain an inner bound for the setup in Example~\ref{exm1}
(i.e., when $\sX=\sY=\{0,1\}$, $\mu=\psi=\sBer(1/2)$, and $\rho$ the Hamming
distance) by restricting the auxiliary random variable $U$ to be
$\sBer(1/2)$. Since $P_X=P_U=P_Y=\sBer(1/2)$, for any $P_{X,Y,U} \in \M(D)$,
the channels $P_{U|X}$ and $P_{Y|U}$ must be $\mathsf{BSC}(a_1)$ and
$\mathsf{BSC}(a_2)$, respectively, for some $a_1, a_2 \in [0,1]$.  Hence, since
$\cE[\rho(X,Y)] = a$ when $P_{X|Y} = \mathsf{BSC}(a)$, the resulting achievable
rate region is
\begin{align}
\L_s(D) &=
\left\{\begin{array}{rcl}
(R,R_c) \in \R^2 &:& (a_1,a_2) \in \Phi(D) \text{  s.t. } \\
 R &\geq &1-h(a_1),  \\
 R+R_c &\geq &1-h(a_2).
\end{array}\right\} \nonumber,
\intertext{where}
\Phi(D) &\coloneqq \{(a_1,a_2) \in [0,1]^2: a_1+a_2-2a_1a_2 \leq D\}. \nonumber
\end{align}

\hspace{-13pt} Let us define $\varphi(a_1,a_2)=a_1+a_2-2a_1a_2$. Note that since $\varphi(\frac{1}{2}+r,\frac{1}{2}+m) = \frac{1}{2} - 2 rm$ and $h(\frac{1}{2}-r)=h(\frac{1}{2}+r)$ for any $r,m \in [\frac{-1}{2},\frac{1}{2}]$; we may assume without loss of generality that $a_1,a_2 \in [0,\frac{1}{2}]$ in the definition of $\Phi(D)$. Furthermore, since $\varphi(a_1,a_2) > D$ when $D<a_1<\frac{1}{2}$ or $D<a_2<\frac{1}{2}$, we can refine the definition of $\L_s(D)$ for $0\leq D < \frac{1}{2}$ as
\begin{align}
\L_s(D) &=
\left\{\begin{array}{rcl}
(R,R_c) \in \R^2 &:& (a_1,a_2) \in \Phi_r(D) \text{  s.t. } \\
 R &\geq &1-h(a_1),  \\
 R+R_c &\geq &1-h(a_2).
\end{array}\right\} \nonumber,
\intertext{where}
\Phi_r(D) &\coloneqq \{(a_1,a_2) \in [0,D]^2: a_1+a_2-2a_1a_2 \leq D\}. \nonumber
\end{align}
Notice that for any fixed $a_1$, $(a_1,a_2) \in \Phi_r(D)$ if and only if $a_2 \leq \frac{D-a_1}{1-2a_1}$, where the expression on the righthand side of the inequality is a concave function of $a_1$. Hence, $\Phi_r(D)$ is a convex region. In the remainder of this section we characterize the boundary $\bigcup_{R_c} \min \{R : (R,R_c) \in \L_s(D)\}\times \{R_c\}$ of $\L_s(D)$.

If $R_c=\infty$, then $(R,\infty) \in \L_s(D) \Leftrightarrow R \geq 1-h(a_1)$ where $a_1 \in [0,D]$. Hence, the minimum $R$ is equal to $1-h(D)$ for $R_c = \infty$. Moreover, if $R=1-h(D)$ or equivalently $a_1 = D$, then $(R,R_c) \in \L_s(D) \Leftrightarrow R_c+1-h(D) \geq 1-h(a_2)=1-h(0)=1$ since $(D,a_2) \in \Phi_r(D)$ only if $a_2=0$. Hence, if and only if $R_c \geq h(D)$, then
\begin{align}
\min\{R: (R,R_c)\in \L_s(D)\}=1-h(D). \nonumber
\end{align}
Note that since $1-h(D)$ is the rate-distortion function for the $\sBer(1/2)$
source and $\psi=\sBer(1/2)$ is the unique output distribution achieving this
rate-distortion function, Proposition~\ref{rdprop} implies that the inner bound we
obtain in this section is tight for $R_c \geq h(D) = H(\psi) - (1-h(D)) =
R_c^{\min}$.

Recall that for an arbitrary $0 \leq R_c < h(D)$, $(R,R_c) \in \L_s(D)
\Leftrightarrow R \geq \max\{1-h(a_1),1-h(a_2)-R_c\}$ where $(a_1,a_2) \in
\Phi_r(D)$. We now prove that the minimum $R$ is attained when
$1-h(a_1)=1-h(a_2)-R_c$ and $a_1+a_2-2a_1a_2=D$. The second equality is clear
since the binary entropy function $h$ is increasing in $[0,D]$. To prove the
first claim by contradiction, let us assume (without loss of generality) that
the minimum is achieved when $1-h(a_1) > 1-h(a_2)-R_c$ $\bigl($so $\min\{R:
(R,R_c) \in \L_s(D)\} = 1-h(a_1)$$\bigr)$. Note that
 $(a_1,a_2) \in \Phi_r(D)$ if and only if $a_2
  \leq \frac{D-a_1}{1- 2 a_1}$, where  $\frac{D-a_1}{1- 2
    a_1}$ is a positive, decreasing,  and concave function of $a_1$ in
  $(0,D)$. This and the fact that  $h$ is
  increasing and continuous imply that  there exist
  $\varepsilon_1,\varepsilon_2>0$ such that
  $(a_1+\varepsilon_1,a_2-\varepsilon_2) \in \Phi_r(D)$ and
  $1-h(a_1+\varepsilon_1) \geq 1-h(a_2-\varepsilon_2)-R_c$. But $\min\{R:
  (R,R_c) \in \L_s(D)\} = 1-h(a_1) > 1-h(a_1+\varepsilon_1)$, which is a
  contradiction.

Hence, for all $D \in (0,\frac{1}{2})$ the minimum coding rate when $0 \leq R_c < h(D)$ is given by
\begin{align}
&\min\{R:(R,R_c) \in \L_s(D)\} \nonumber \\
&\phantom{xxxxxxxxxxx}=\min\{1-h(a_1): (a_1,a_2) \in \Pi(D,R_c) \} \nonumber \\
\intertext{where}
&\Pi(D,R_c) \nonumber \\
&\phantom{xxx}\coloneqq\left\{\hspace{-5pt}\begin{array}{rl}
(a_1,a_2) \in \Phi_r(D) : 1-h(a_1) = 1-h(a_2)-R_c& \\
\text{  and } a_1+a_2-2a_1a_2 =D \phantom{x}&
\end{array}\hspace{-15pt}\right\}. \nonumber
\end{align}



\begin{figure}[h]
\centering
\includegraphics[width=3.8in, height=2.5in]{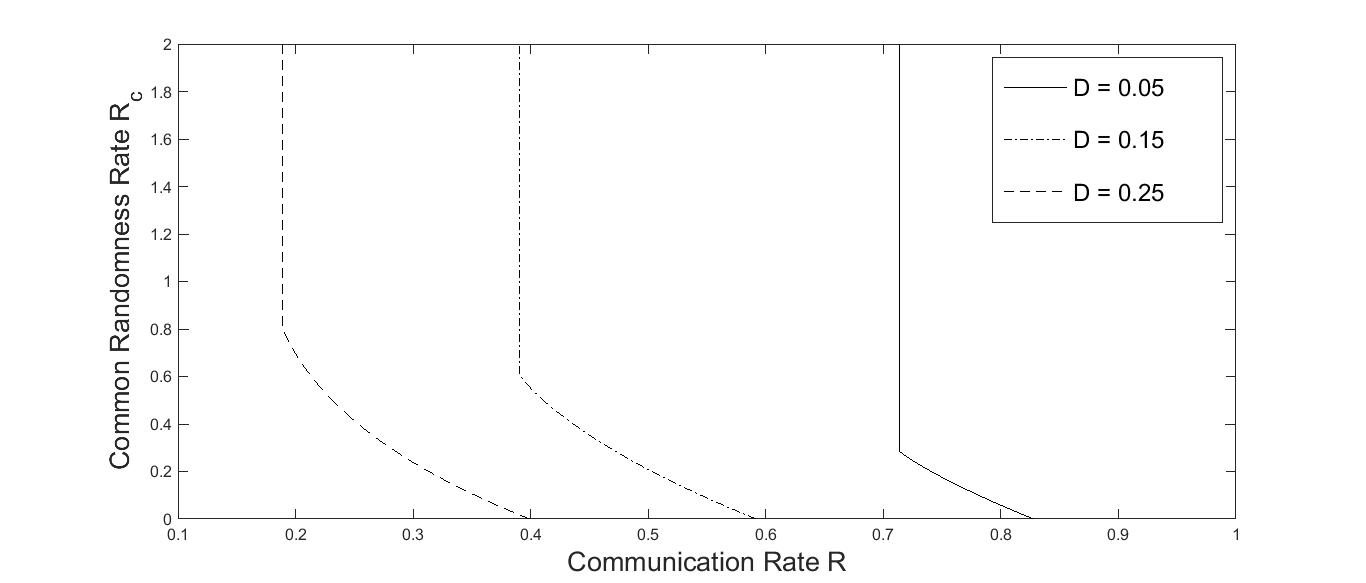}
\caption{$\L_s(D)$ for binary symmetric source at different distortion levels
  $D$. }
\label{gr1}
\end{figure}

Figure~\ref{gr1} shows the rate region $\L_s(D)$ for $D=0.25$,
  $D=0.15$, and $D=0.05$. At the boundary of $\L_s(D)$, the coding rate $R$
  ranges from $1-h(a^{*})=0.39, 0.59, 0.82$ bits
  $\bigl($$a^{*}=\frac{1}{2}(1-\sqrt{1-2D})$$\bigr)$ to $h(D)=0.19, 0.4, 0.72$
  bits, respectively, while the common randomness rate $R_c$ ranges from $0$ to
  $1-h(D)=0.81, 0.6, 0.28$  for $D = 0.25$, $D = 0.15$, and
  $D = 0.05$, respectively.

\subsection{Gaussian Source}

Let $\N(m,\sigma)$ denote a Gaussian random variable with mean $m$ and variance $\sigma^2$ (similar notation will be used for the vector case). In this section, we obtain an inner bound for the case $\sX = \sY = \R$, $\mu = \N(0,\sigma_X)$, $\psi = \N(0,\sigma_Y)$, and $\rho$ is the squared error distortion (i.e., $\rho(x,y) = |x-y|^2$) by restricting $(X,U,Y)$ to be Gaussian $\bigl($or, equivalently, restricting $(X,U)$ and $(U,Y)$ to be Gaussian since $X - U - Y$$\bigr)$.

\begin{remark}\label{rmk1}
Recall that for $R_c=\infty$, the minimum coding rate is given by (\ref{eq_lower_mutual}). However if $X \sim \N(0,\sigma_X)$ and $Y \sim \N(0,\sigma_Y)$, then for any $P_{X,Y} \in \G(D)$, one has the lower bound
\begin{align}
&I(X;Y) = h(X) + h(Y) - h(X,Y) \nonumber \\
&\phantom{xx}\geq \frac{1}{2} \log(2\pi e \sigma_X^2) + \frac{1}{2} \log(2 \pi e \sigma_Y^2) - \log(2 \pi e \cdet(C)^{\frac{1}{2}}), \nonumber \end{align}
where $C$ is the covariance matrix of $(X,Y)$. The equality is achieved when $(X,Y)$ is jointly Gaussian \cite[Theorem 8.6.5]{CoTh06}. Hence, we can restrict $(X,Y)$ to be Gaussian in the definition of $I(\mu\|\psi,D)$, i.e.,
\begin{align}
I(\mu\|\psi,D) &\coloneqq \min\{ I(X,Y): P_{X,Y}\in \G_g(D)\}, \nonumber \\
\intertext{where}
\G_g(D)&\coloneqq \{P_{X,Y} \in \G(D): P_{X,Y} = \N(0,C) \text{ for some } C\}. \nonumber
\end{align}
This implies that the inner bound we obtain in this section is tight for $R_c=\infty$ $\bigl($i.e., $\L_s(D,\infty)=\L(D,\infty)$$\bigr)$. $\L(D,\infty)$ for the case $\mu=\psi=\N(0,\sigma)$ was derived in \cite[Proposition 2]{LiKlKl11}.
\end{remark}

Note that without loss of generality we can take $U$ to have zero mean and unit
variance. Indeed, let $\tilde{U} = (U-\delta_U)/\sigma_U$. Then $\tilde{U} \sim
\N(0,1)$, $X-\tilde{U}-Y$, and $(X,\tilde{U},Y)$ is Gaussian with $I(X;U) =
I(X;\tilde{U})$ and $I(Y;U) = I(Y;\tilde{U})$. Hence, in the remainder of this
section, we assume $U \sim \N(0,1)$.


Let us write $U=aX+V$ and $Y=bU+W$, where $a,b \in \R$, and $V \sim \N(0,\sigma_V)$, $W \sim \N(0,\sigma_W)$, and $(X,V,W)$ are independent. With this representation, the constraints in the definition of the achievable rate region become
\begin{align}
1 &= a^2 \sigma_X^2 + \sigma_V^2, \nonumber \\
\sigma_Y^2 &= b^2+\sigma_W^2, \nonumber \\
(1-ab)^2\sigma_X^2+b^2\sigma_V^2+\sigma_W^2 &\leq D, \nonumber
\end{align}
Then, if we substitute $\sigma_V^2=1-a^2 \sigma_X^2 \geq 0$ and $\sigma_W^2=\sigma_Y^2-b^2 \geq 0$ into the last equation, we can
write the distortion constraint as
\begin{align}
\sigma_X^2+\sigma_Y^2-2ab \sigma_X^2\leq D. \nonumber
\end{align}
Since
\begin{align}
I(X;U) &= H(X) + H(U) - H(X,U) \nonumber \\
&= \frac{1}{2} \log(2\pi\e \sigma_X^2) + \frac{1}{2} \log(2\pi\e) - \log(2\pi\e\cdet(C_X)^{\frac{1}{2}}) \nonumber \\
&= \frac{1}{2}\log\bigl(\frac{1}{(1-a^2\sigma_X^2)}\bigr) \nonumber \\
\intertext{and}
I(Y;U) &= H(Y) + H(U) - H(Y,U) \nonumber \\
&= \frac{1}{2} \log(2\pi\e \sigma_Y^2) + \frac{1}{2} \log(2\pi\e) - \log(2\pi\e\cdet(C_Y)^{\frac{1}{2}}) \nonumber \\
&= \frac{1}{2}\log\bigl(\frac{\sigma_Y^2}{(\sigma_Y^2-b^2)}\bigr), \nonumber
\end{align}
where $C_X$ is the covariance matrix of $(X,U)$ and $C_Y$ is the covariance matrix of $(Y,U)$, the resulting achievable rate region can be written as
\begin{align}
&\L_s(D) =
\left\{\begin{array}{rcl}
(R,R_c) \in \R^2 &:& (a,b) \in \Psi(D) \text{  s.t. } \\
 R &\geq &\frac{1}{2} \log\bigl(\frac{1}{(1-a^2\sigma_X^2)}\bigr),  \\
 R+R_c &\geq &\frac{1}{2} \log\bigl(\frac{\sigma_Y^2}{(\sigma_Y^2-b^2)}\bigr).
\end{array}\right\} \nonumber,
\intertext{where}
&\Psi(D) \nonumber \\
&\phantom{xx}\coloneqq \{(a,b) \in [0,\sigma_X^{-1}]\times [0,\sigma_Y]: \sigma_X^2+\sigma_Y^2-2ab \sigma_X^2\leq D\}.\nonumber
\end{align}
Note that the region $\Psi(D)$ is convex. Let us define $I_1(a)=\log\bigl(\frac{1}{(1-a^2\sigma_X^2)}\bigr)$ and $I_2(b)=\log\bigl(\frac{\sigma_Y^2}{(\sigma_Y^2-b^2)}\bigr)$; then $I_1$ and $I_2$ are increasing functions.
As in Section~\ref{ex1}, we characterize the boundary $\bigcup_{R_c} \min \{R : (R,R_c) \in \L_s(D)\}\times \{R_c\}$ of $\L_s(D)$.

If $R_c=\infty$, then $(R,\infty) \in \L_s(D) \Leftrightarrow R \geq I_1(a)$ where $(a,b) \in [0,\sigma_X^{-1}]\times [0,\sigma_Y]$ and  $\sigma_X^2+\sigma_Y^2-2ab \sigma_X^2\leq D$. Using the monotonicity of $I_1$ and the distortion constraint, it is straightforward to show that
\begin{align}
\min\{R: (R,\infty) \in \L_s(D)\} = I_1\bigl(\frac{\sigma_X^2+\sigma_Y^2-D}{2\sigma_X^2\sigma_Y}\bigr).\nonumber
\end{align}
By Remark~\ref{rmk1}, this is the minimum coding rate (i.e., rate-distortion function) for $R_c=\infty$.

When $0 \leq R_c < \infty$ is arbitrary, we can use the same technique as in Section~\ref{ex1} to prove that the minimum of $R$ is attained when $I_1(a) = I_2(b) - R_c$ and $\sigma_X^2+\sigma_Y^2-2ab \sigma_X^2=D$  ($I_1$ and $I_2$ are increasing continuous functions and $\Psi(D)$ is a convex region with nonempty interior in the upper-right corner of the rectangle $[0,\sigma_X^{-1}]\times [0,\sigma_Y]$). As a consequence, we can describe the minimum coding rate when $0 \leq R_c < \infty$ as follows:
\begin{align}
&\min\{R:(R,R_c) \in \L_s(D)\}=\min \{I_1(a): (a,b) \in \Lambda(D,R_c) \} \nonumber \\
\intertext{where}
&\Lambda(D,R_c)
\coloneqq\left\{\hspace{-5pt}\begin{array}{rl}
(a,b) \in \Psi(D) :& I_1(a) = I_2(b)-R_c \text{  and }  \\
&\sigma_X^2+\sigma_Y^2-2ab \sigma_X^2=D \phantom{x}
\end{array}\hspace{-10pt}\right\}. \nonumber
\end{align}

\begin{figure}[h]
\centering
\includegraphics[width=3.8in, height=2.5in]{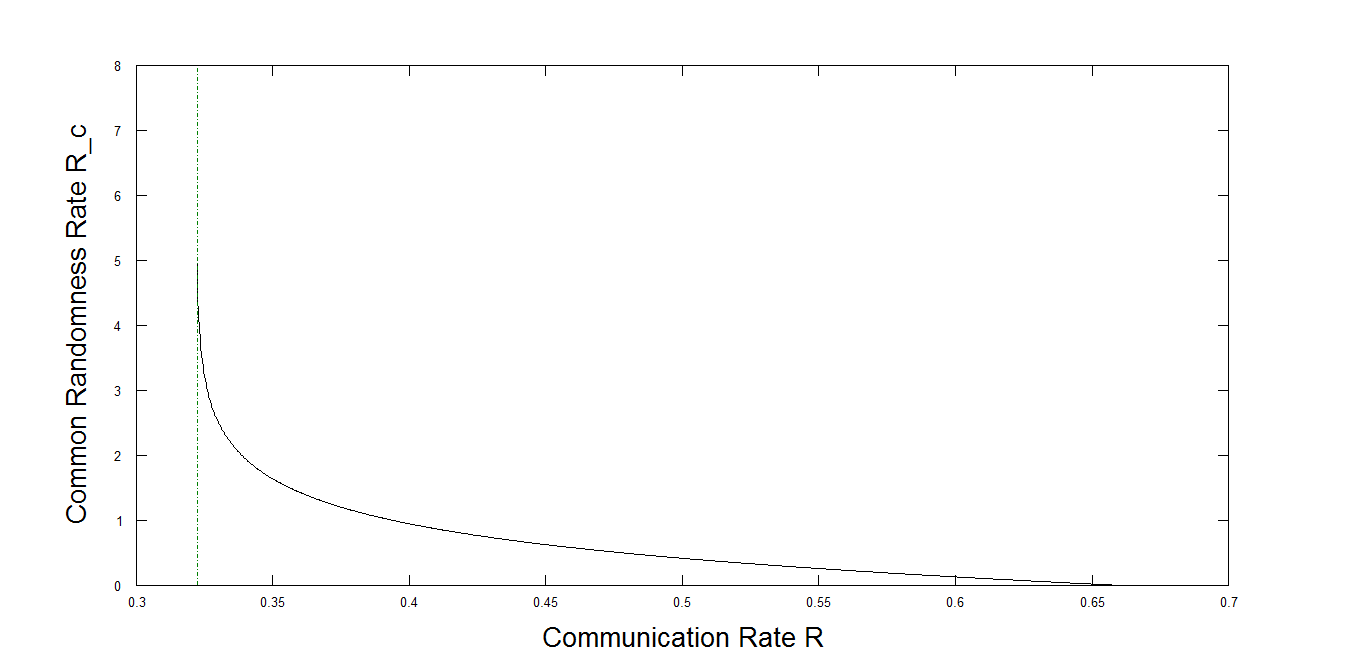}
\caption{$\L_s(D)$ for Gaussian source for $D=0.8$.}
\label{gr2}
\end{figure}

Figure~\ref{gr2} shows the rate region $\L_s(D)$ for $\sigma_X = \sigma_Y = 1$ and $D=0.8$. At the boundary of $\L_s(D)$, the coding rate $R$ ranges
from $I_1(\sqrt{\frac{2-D}{2}})=0.65$ bits to $I_1(\frac{2-D}{2})=0.32$ bits
while the common randomness rate $R_c$ ranges from $0$ to infinity.

\section{TWO VARIATIONS}
\label{sec3}

In this section we consider two variations of the rate-distortion problem defined in Section~\ref{sec2}. Throughout this section we assume that the source alphabet $\sX$ and the reproduction alphabet $\sY$ are finite.

\subsection{Rate Region with Empirical Distribution Constraint}
\label{sec3sub1}

First, we investigate the effect on the achievable rate region of relaxing the strict output distribution constraint on $Y^n$ and requiring only that the empirical output distribution $p_{Y^n}$ converges to the distribution $\psi$.

\begin{definition}
\label{def3}
For any positive real number $D$ and desired output distribution $\psi$, the pair $(R,R_c)$ is said to be \emph{empirically achievable} if there exists a sequence of $(n,R,R_c)$ randomized source codes such that
\begin{align}
\limsup_{n\rightarrow\infty} \cE[\rho_n(X^n,Y^n)] &\leq D,\nonumber \\
\|p_{Y^n}-\psi\|_{TV} &\rightarrow 0 \text{  in probability as } n\rightarrow\infty. \nonumber
\end{align}
\end{definition}

For any $D\geq0$ we let $\cR_e(D)$ denote the set of all empirically achievable rate pairs $(R,R_c)$, and define $\cR_e(D,R_c)$ as the set of coding rates $R$ such that $(R,R_c) \in \cR_e(D)$.

This setup is motivated by the work of Cuff \textit{et. al.} \cite[Section II]{Cuf10} on empirical coordination. The main objective of \cite[Section II]{Cuf10} is to empirically simulate a memoryless channel by a system as in Fig.~\ref{fig1}. To be more precise, let $Q(y|x)$ denote a given  discrete memoryless channel with
input alphabet $\sX$ and output alphabet $\sY$ to  be simulated
(synthesized) for input $X$ having distribution $\mu$. Let $\pi=\mu Q$ be the
joint distribution of the resulting input-output pair $(X,Y)$.
\begin{definition}
\label{def5}
The pair $(R,R_c)$ is said to be \emph{achievable} for empirically synthesizing a memoryless channel $Q$ with input distribution $\mu$ if there exists a sequence of $(n,R,R_c)$ randomized source codes such that
\begin{align}
\lim_{n\rightarrow\infty} \|p_{X^n,Y^n} - \pi\|_{TV} = 0 \text{ in probability}. \label{eq36}
\end{align}
\end{definition}

Let $\C_e$ denote the the set of all achievable $(R,R_c)$ pairs and let $\C_e(R_c)$ denote the set of all rates $R$ such that $(R,R_c) \in \C_e$. The following theorem, which is a combination of \cite[Theorems 2 and 3]{Cuf10}, characterizes the entire set $\C_e$.

\begin{theorem}
\label{thm6}
The set $\C_e$ of all achievable $(R,R_c)$ is given by
\begin{align}
\C_e =& \left\{ \begin{array}{rcl} (R,R_c)\in\R^2 &:& \exists P_{X,Y}\in\G \text{ s.t.} \\
 R&\geq &I(X;Y)\end{array} \right\} \nonumber,
\end{align}
where
\begin{align}
  \G \coloneqq& \{P_{X,Y}: P_{X,Y} = \pi\}. \nonumber
\end{align}
Hence, $\C_e(R_c) = \C_e(0)$ for any $R_c$.
\end{theorem}
Using the above theorem and the arguments in \cite[Section VII]{Cuf10}, one can show that the set of empirically achievable rate pairs $(R,R_c)$ at the distortion level $D$ can be described as:

\begin{theorem}
\label{thm3}
For any $D\geq0$ we have
\begin{align}
\cR_e(D,0)   &= \L(D,\infty), \nonumber \\
\cR_e(D,R_c) &= \cR_e(D,0) \text{  for all } R_c. \label{eq18}
\end{align}
In other words, $\cR_e(D) = \L(D,\infty) \times [0,\infty)$.
\end{theorem}

The proof of Theorem~\ref{thm3} is given in Appendix~\ref{sec6sub2a}.
Note that (\ref{eq18}) states that unlike in the original problem defined in Section~\ref{sec2}, here common randomness cannot decrease the necessary coding rate.

\subsection{Deterministic-Decoder Rate Region}
\label{sec3sub2}

In this section we investigate the effect on the rate region of private
randomness used by the decoder. Namely, we determine the achievable rate region
for a randomized source code having no (private) randomness at the decoder,
i.e., when the decoder $F$ is a deterministic function of random variables $J$
and $K$. We call such a code a \emph{randomized source code with
    deterministic decoder}. In this setup, since the encoder can reconstruct
the output $Y^n$ of the decoder by reading off $J$ and $K$, the common
randomness $K$ may be interpreted as feedback from the output of the decoder to
the encoder \cite[p.\ 5]{Win02}.

\begin{definition}
  For any positive real number $D$ and desired output distribution $\psi$, the
  pair $(R,R_c)$ is said to be \emph{achievable with a
      deterministic decoder} if there exists a sequence of $(n,R,R_c)$
  randomized source codes with  a deterministic decoder such
  that
\begin{align}
&\limsup_{n\rightarrow\infty}\cE[\rho_n(X^n,Y^n)] \leq D,\nonumber \\
&\lim_{n\rightarrow\infty}\|P_{Y^n}-\psi^n\|_{TV} = 0. \label{eq20}
\end{align}
\end{definition}

Note that here we relax the strict i.i.d.\ output distribution
  constraint, because without private randomness at the decoder, some output
  distributions cannot be exactly achieved for finite rates $(R,R_c)$. Indeed,
  this is the case when the probabilities of the
  output distribution are irrational and the input distribution
  has rational probabilities.

  For any $D\geq0$ we let $\cR_{dd}(D)$ denote the set of all achievable
  $(R,R_c)$ pairs with deterministic decoder. The following theorem, proved in
  Appendix~\ref{sec6sub3}, characterizes the closure of this set.

\begin{theorem}
\label{thm4}
For any $D\geq0$,
\begin{align}
\cl\cR_{dd}(D) &= \left\{ \begin{array}{rcl}
(R,R_c)\in\R^2 &:& \exists P_{X,Y} \in \G(D) \text{ s.t.} \\
R &\geq &I(X;Y), \\
R+R_c &\geq &H(Y)\end{array}\right\}. \label{eq21}
\end{align}
\end{theorem}

\begin{remark}
\begin{itemize}
\item [ ]
\item[(a)] Note that the rate region in Theorem~\ref{thm4} can equivalently be given by
\begin{align}
&\cl\cR_{dd}(D) \nonumber\\
&\phantom{xx}= \left\{ \begin{array}{rcl} (R,R_c)\in\R^2 &:& \exists P_{X,Y,U} \in \M(D) \text{ s.t.} \\
R &\geq &I(X;U), \\
R+R_c &\geq &H(Y) \end{array} \right\}. \label{eq22}
\end{align}
Therefore, $\L(D) \supset \cl\cR_{dd}(D)$.
\item[(b)] It is important to note that if we allow the decoder to use private randomness while preserving the output distribution constraint (\ref{eq20}), one can prove that the resulting achievable rate region is $\L(D)$. In this case, the only part to prove is the converse, since the achievability is obvious. However, the converse can be proven by using a similar technique as in \cite[Section VI]{Cuf13}. Hence, if we allow the decoder to use private randomness, replacing the strict output distribution constraint in the Definition~\ref{def1} with (\ref{eq20}) does not change the achievable rate region.

\item[(c)] Since $\L(D) \supset \cl\cR_{dd}(D)$, where the inclusion is strict
  in general, private randomness can indeed replaces a part of the common
  randomness to decrease the necessary coding rate  when the common randomness
  rate is less than $R_c^{\min}$.
\end{itemize}
\end{remark}

\section{Proof of Theorem~\ref{thm2}}
\label{sec_proof}

Our proof relies on techniques developed by Cuff in \cite{Cuf13}. In particular, in the achievability part, we apply the `likelihood encoder' of \cite{Cuf10,Cuf13} which is an elegant alternative to the standard random coding argument. The converse part of the proof is an appropriately modified version of the converse argument in \cite{Cuf13}; however, in our setup this technique also works in the continuous alphabet case, while in \cite{Cuf13} the finite alphabet assumption seem quite difficult to relax.

\subsection{Achievability for Discrete Alphabets}
\label{sec achiev}

Assume that $(R,R_c)$ is in the interior of $\L(D)$. Then there exists $P_{X,Y,U} \in \M(D)$ such that $R>I(X;U)$ and $R+R_c>I(Y;U)$. The method used in this part of the proof comes from \cite[Section V]{Cuf13} where instead of explicitly constructing the encoder-decoder pair, a joint distribution was constructed from which the desired encoder-decoder behavior is established.

In this section, distributions which depend on realizations of some random variable (e.g., random codebook) will be denoted as bold upper case letters, but without referring to the corresponding realization for notational simplicity.

For each $n$, generate a random `codebook' $\C_n \coloneqq \bigl\{U^n(j,k)\bigr\}$ of $u^n$ sequences independently drawn from $P_{U}^n$ and indexed by $(j,k)\in [2^{nR}] \times [2^{nR_c}]$. For each realization $\{u^n(j,k)\}$ of $\C_n$, define a distribution $\mathbf{\Gamma}_{X^n,Y^n,J,K}$ such that $(J,K)$ is uniformly distributed on $[2^{nR}] \times [2^{nR_c}]$ and $(X^n,Y^n)$ is the output of the stationary and memoryless channel $P_{X,Y|U}^n$ when we feed it with $u^n(J,K)$, i.e.,
\begin{align}
\mathbf{\Gamma}_{X^n,Y^n,J,K}(x^n,y^n,j,k) \coloneqq \frac{1}{2^{n(R+R_c)}} P_{X,Y|U}^n(x^n,y^n|u^n(j,k)). \label{aeq1}
\end{align}
Here, $\{\mathbf{\Gamma}_{X^n,Y^n,J,K}\}_{n\geq1}$ are the distributions from which we derive a sequence of encoder-decoder pairs which for all $n$ large enough \emph{almost} meet the requirements in Definition~\ref{def1}.

\begin{lemma}[Soft covering lemma{ \cite[Lemma IV.1]{Cuf13}}]
Let $P_{V,W}=P_V P_{W|V}$ be the joint distribution of some random vector $(V,W)$ on $\sV \times \sW$, where $P_V$ is the marginal on $\sV$ and $P_{W|V}$ is the conditional probability on $\sW$ given $\sV$. For each $n$, generate the set $\B_n = \bigl\{V^n(i)\bigr\}$ of $v^n$ sequences independently drawn from $P_V^n$ and indexed by $i \in [2^{nR}]$. Let us define a random measure on $\sW^n$ as
\begin{align}
\mathbf{P}_{W^n}(w^n) \coloneqq \frac{1}{|\B_n|} \sum_{i=1}^{|\B_n|} P_{W^n|V^n}(w^n|V^n(i)), \nonumber
\end{align}
where $P_{W^n|V^n} = \prod_{i=1}^n P_{W|V}$. If $R \geq I(V;W)$, then we have
\begin{align}
\cE_{\B_n} \bigl[\| \mathbf{P}_{W^n} - P_W^n \|_{TV} \bigr] \leq \frac{3}{2} \exp\{- \kappa n\}, \nonumber
\end{align}
for some $\kappa > 0$.
\end{lemma}

Since $R+R_c > I(Y;U)$, by the soft covering lemma
\begin{align}
\cE_{\C_n}\bigl[\|\mathbf{\Gamma}_{Y^n} - P_{Y}^n\|_{TV}\bigr] &\leq \frac{3}{2} \exp{\{-cn\}} \label{aeq3},
\end{align}
where $c> 0$ and $\cE_{\C_n}$ denotes expectation with respect to the distribution of $\C_n$. Note that for any fixed $k$, the collection $\C_n(k) \coloneqq \{U^n(j,k)\}_j$ is a random codebook of size $2^{nR}$. Since $R > I(X;U)$, the soft covering lemma again gives
\begin{align}
\cE_{\C_n(k)} \bigl[\| \mathbf{\Gamma}_{X^n|K=k} - P_{X}^n  \|_{TV} \bigr] \leq \frac{3}{2} \exp{\{-dn\}} \label{aeq4},
\end{align}
where $d > 0$ (same for all $k$) and $\cE_{\C_n(k)}$ denotes expectation with respect to the distribution of $\C_n(k)$. Then, by the definition of total variation, we have
\begin{align}
&\cE_{\C_n}\bigl[ \|\mathbf{\Gamma}_{X^n,K} - \frac{1}{2^{nR_c}}P_{X}^n\|_{TV} \bigr] \nonumber \\
& \phantom{xxxxx}\coloneqq\cE_{\C_n}\biggl[ \frac{1}{2}\sum_{x^n,k} \bigl| \mathbf{\Gamma}_{X^n,K}(x^n,k) - \frac{1}{2^{nR_c}} P_{X}^n(x^n) \bigr| \biggr] \nonumber \\
&\phantom{xxxxx} =\frac{1}{2^{nR_c}} \cE_{\C_n}\biggl[ \frac{1}{2}\sum_{x^n,k} \bigl| \mathbf{\Gamma}_{X^n|K}(x^n|k) - P_{X}^n(x^n) \bigr| \biggr]  \nonumber \\
&\phantom{xxxxx}= \frac{1}{2^{nR_c}} \sum_k \cE_{\C_n(k)}\bigl[ \| \mathbf{\Gamma}_{X^n|K=k} - P_{X}^n \|_{TV} \bigr] \nonumber \\
&\phantom{xxxxx} \leq\frac{3}{2} \exp{\{-dn\}}. \label{aeq5}
\end{align}

Furthermore, the expected value (taken with respect to the distribution of $\C_n$) of the distortion induced by $\mathbf{\Gamma}_{X^n,Y^n}$ is upper bounded by $D$ as a result of the symmetry in the construction of $\C_n$, i.e.,
\begin{align}
&\cE_{\C_n}\biggl[\sum_{x^n,y^n} \rho_n(x^n,y^n) \mathbf{\Gamma}_{X^n,Y^n}(x^n,y^n)\biggr] \nonumber \\
&\phantom{xxx}=\cE_{\C_n}\biggl[\sum_{j,k} \sum_{x^n,y^n} \rho_n(x^n,y^n) \mathbf{\Gamma}_{X^n,Y^n,J,K}(x^n,y^n,j,k)\biggr] \nonumber \\
&\phantom{xxx}=\sum_{x^n,y^n} \rho_n(x^n,y^n) \sum_{j,k} \cE_{\C_n}\biggl[\mathbf{\Gamma}_{X^n,Y^n,J,K}(x^n,y^n,j,k)\biggr] \nonumber \\
&\phantom{xxx}=\sum_{x^n,y^n} \rho_n(x^n,y^n) P_{X,Y}^n(x^n,y^n) \leq D, \label{aeq20}
\end{align}
where the last equality follows from the symmetry and the independence in the codebook construction, and the last inequality follows from the definition of $\M(D)$.

Now, since $\mathbf{\Gamma}_{Y^n,J|X^n,K} = \mathbf{\Gamma}_{J|X^n,K} \mathbf{\Gamma}_{Y^n|J,K}$, we define a randomized  $(n,R,R_c)$ source code such that it has the encoder-decoder pair $(\mathbf{\Gamma}_{J|X^n,K},\mathbf{\Gamma}_{Y^n|J,K})$.
Hence, $(n,R,R_c)$ depends on the realization of $\C_n$. Let $\mathbf{P}_{X^n,Y^n,J,K}$ denote the distribution induced by $(n,R,R_c)$, i.e.,
\begin{align}
&\mathbf{P}_{X^n,Y^n,J,K}(x^n,y^n,j,k) \nonumber \\
&\phantom{xxxxxxxxxxxx}\coloneqq \frac{1}{2^{nR_c}} P_X^n(x^n) \mathbf{\Gamma}_{Y^n,J|X^n,K}(y^n,j|x^n,k). \nonumber
\end{align}
If two distributions are passed through the same channel, then the total variation between the joint distributions is the same as the total variation between the input distributions \cite[Lemma V.2]{Cuf13}. Hence, by (\ref{aeq5})
\begin{align}
\cE_{\C_n} \biggl[ \|\mathbf{\Gamma}_{X^n,Y^n,K,J} - \mathbf{P}_{X^n,Y^n,K,J} \|_{TV} \biggr] \leq \frac{3}{2} \exp{\{-dn\}} \label{aeq6}.
\end{align}
Then, (\ref{aeq20}) and (\ref{aeq6}) give
\begin{align}
\cE_{\C_n} \biggl[ \sum_{x^n,y^n} \rho_n(x^n,y^n) \mathbf{P}_{X^n,Y^n}(x^n,y^n) \biggr] &\leq D + \alpha \exp{\{-dn\}}, \label{aeq7}
\end{align}
where $\alpha = \rho_{\max} \frac{3}{2}$. By virtue of the properties of total variation distance, (\ref{aeq3}) and (\ref{aeq6}) also imply
\begin{align}
&\cE_{\C_n} \bigl[ \| \mathbf{P}_{Y^n} - P_{Y}^n \|_{TV} \bigr] \nonumber \\
&\phantom{xxxxx}\leq  \cE_{\C_n} \bigl[ \| \mathbf{P}_{Y^n} - \mathbf{\Gamma}_{Y^n} \|_{TV} \bigr] + \cE_{\C_n} \bigl[ \| \mathbf{\Gamma}_{Y^n} - P_{Y}^n \|_{TV} \bigr] \nonumber \\
&\phantom{xxxxx}\leq \frac{3}{2} \exp{\{-dn\}} + \frac{3}{2} \exp{\{-cn\}}  \nonumber \\
&\phantom{xxxxx}= \alpha_n \exp{\{-dn\}} \label{aeq8},
\end{align}
where (without any loss of generality) we assumed  $d < c$ and where $\alpha_n \coloneqq \frac{3}{2}\bigl(1+\exp{\{-(c-d)n\}}\bigr) \leq 2$ if $n$ is large enough.

Define the following functions of the random codebook $\C_n$:
\begin{align}
D(\C_n) &\coloneqq \sum_{x^n,y^n} \rho_n(x^n,y^n) \mathbf{P}_{X^n,Y^n}(x^n,y^n), \nonumber \\
G(\C_n) &\coloneqq \|\mathbf{P}_{Y^n} - P_{Y}^n\|. \nonumber
\end{align}
Thus, the expectations of $D(\C_n)$ and $G(\C_n)$ satisfy (\ref{aeq7}) and (\ref{aeq8}), respectively. For any $\delta \in (0,d)$, Markov's inequality gives
\begin{align}
\sPr\biggr\{ G(\C_n) \leq \exp{\{-\delta n\}} \biggl\} &\geq 1-\frac{\alpha_n \exp{\{-dn\}}}{\exp{\{-\delta n\}}}, \label{aeq9} \\
\sPr\biggr\{ D(\C_n) \leq D+\delta \biggl\} &\geq 1-\frac{D+\alpha \exp{\{-dn\}}}{D+\delta} \label{aeq10}.
\end{align}
Since
\begin{align}
&\lim_{n\rightarrow\infty} \biggl(2 - \frac{\alpha_n \exp{\{-dn\}}}{\exp{\{-\delta n\}}} - \frac{D+\beta \exp{\{-dn\}}}{D+\delta} \biggl) \nonumber \\
&\phantom{xxxxxxxxxxxxxxxxxx}= 2 - \frac{D}{D+\delta} > 1, \nonumber
\end{align}
there exists a positive $N(\delta)$ such that for $n \geq N(\delta)$,  we have
\begin{align}
\sPr \biggl\{ \biggl(D(\C_n) \leq D+\delta\biggr) \bigcap \biggl(G(\C_n) \leq \exp{\{-\delta n\}}\biggr) \biggr\} > 0. \nonumber
\end{align}
This means that for each $n\geq N(\delta)$, there is a realization of $\C_n$ which gives
\begin{align}
\sum_{x^n,y^n} \rho_n(x^n,y^n) \mathbf{P}_{X^n,Y^n}(x^n,y^n) &\leq D+\delta \label{aeq11} \\
\|\mathbf{P}_{Y^n} - P_{Y}^n\| &\leq \exp{\{-\delta n\}} \label{aeq12}.
\end{align}
Hence, the sequence of $(n,R,R_c)$ randomized source codes corresponding to these realizations almost satisfies the achievability constraints. Next we can slightly modify this coding scheme so that the code exactly satisfies the i.i.d. output distribution constraint $Y^n = \psi^n = P_Y^n$ while having distortion upper bounded by $D+\delta$.

Before presenting this modification, we pause to define the notion of optimal coupling and the optimal transportation cost as they will play an important role in the sequel. Let $\pi$, $\lambda$ be probability measures over finite or continuous alphabets $\sW$ and $\sV$, respectively. The optimal transportation cost $\hat{T}(\pi,\lambda)$ between $\pi$ and $\lambda$ (see, e.g., \cite{Vil09}) with respect to a cost function $c:\sV\times\sW\rightarrow[0,\infty)$ is defined by
\begin{align}
\hat{T}(\pi,\lambda) = \inf\bigl\{\cE[c(V,W)]: V \sim \pi, W \sim \lambda \bigr\}, \label{neq14}
\end{align}
where the infimum is taken over all joint distribution of pairs of random variables $(V,W)$ satisfying the given marginal distribution constraints. The distribution achieving $\hat{T}(\pi,\lambda)$ is called an optimal coupling of $\pi$ and $\lambda$. Somewhat informally, we also call the corresponding conditional probability on $\sW$ given $\sV$ an optimal coupling. Optimal couplings exist when $\sV=\sW$ are finite or when $\sV=\sW=\R$, $\rho(x,y)=(x-y)^2$, and both $\pi$ and $\lambda$ both have finite second moments \cite{Vil09}.

 Consider the $(n,R,R_c)$ randomized source code depicted in
  Fig.~\ref{fig4} which is obtained by augmenting the original $(n,R,R_c)$ code
  with the optimal coupling $T_{\hat{Y}^n|Y^n}$ between $\mathbf{P}_{Y^n}$ and
  $\psi^n$ with transportation cost $\hat{T}(\mathbf{P}_{Y^n},\psi^n)$ when the
  cost function is $\rho_n(y^n,\tilde{y}^n) = \frac{1}{n} \sum_{i=1}^n
  \rho(y_i,\tilde{y}_i) \coloneqq \frac{1}{n} \sum_{i=1}^n d(y_i,\tilde{y}_i)^p$,
  where $d$ is a metric on $\sY = \sX$. Note that
\begin{align}
d_n(y^n,\tilde{y}^n) \coloneqq \biggl(\sum_{i=1}^n d(y_i,\tilde{y}_i)^p\biggr)^{\frac{1}{q}}, \nonumber
\end{align}
where $q = \max\{1,p\}$, defines a metric on $\sY^n$. We have
\begin{eqnarray*}
\lefteqn{ \hat{T}(\mathbf{P}_{Y^n},\psi^n)}\\
& \coloneqq &\!\!\! \inf\biggl\{\cE\bigl[\rho_n(Y^n,\tilde{Y}^n)\bigr]: Y^n \sim \mathbf{P}_{Y^n}, \tilde{Y}^n \sim \psi^n \biggr\} \nonumber \\
&= &  \!\!\! \frac{1}{n} \inf\biggl\{\cE\biggl[\sum_{i=1}^n d(Y_i,\tilde{Y}_i)^p\biggr]: Y^n \sim \mathbf{P}_{Y^n}, \tilde{Y}^n \sim \psi^n \biggr\} \nonumber \\
&= &  \!\!\! \frac{1}{n} \biggl( \inf\biggl\{\cE\biggl[\sum_{i=1}^n d(Y_i,\tilde{Y}_i)^p\biggr]^{\frac{1}{q}}: Y^n \sim \mathbf{P}_{Y^n}, \tilde{Y}^n \sim \psi^n \biggr\}\biggr)^{q} \nonumber \\
&= & \!\!\! \frac{1}{n} \bigl( W_q(\mathbf{P}_{Y^n},\psi^n) \bigr)^q, \nonumber
\end{eqnarray*}
where $W_q$ denotes the Wasserstein distance of order $q$ \cite[Definition 6.1]{Vil09}.
\begin{figure}[h]
\centering
\tikzstyle{int}=[draw, fill=white!20, minimum size=3em]
\tikzstyle{init} = [pin edge={to-,thin,black}]
\scalebox{0.7}{
\begin{tikzpicture}[node distance=3cm,auto,>=latex']
    \node [int] (a) {$\mathbf{\Gamma}_{J|X^n,K}$};
    \node (b) [left of=a,node distance=3cm, coordinate] {a};
    \node [int] (c) [right of=a] {$\mathbf{\Gamma}_{Y^n|J.K}$};
    \node [int] (d) [right of=c,node distance=3cm] {$T_{\hat{Y}^n|Y^n}$};
    \node [coordinate] (end) [right of=d, node distance=3cm]{};
    \path[->] (b) edge node {$X^n\sim\mu^n$} (a);
    \path[-] (a) edge node {$J$} (c);
    \path[-] (c) edge node {$Y^n$} (d);
    \draw[->] (d) edge node {$\hat{Y}^n\sim\psi^{n}$} (end) ;
\end{tikzpicture}}
\caption{Randomized source code used in the achievability
    proof  for discrete alphabets.}
\label{fig4}
\end{figure}
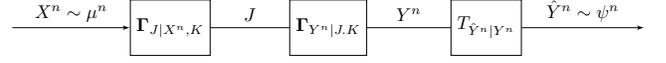
Using \cite[Theorem 6.15]{Vil09}, we obtain for arbitrary fixed $y_0^n \in \sY^n$ and
$r$  such that  $\frac{1}{q} + \frac{1}{r} = 1$,
\begin{align}
W_{q}(\mathbf{P}_{Y^n},\psi^n) &\leq 2^{\frac{1}{r}} \biggl( \sum_{y^n} d_n(y_0^n,y^n)^q \bigl|\mathbf{P}_{Y^n}(y^n) - \psi^n(y^n)\bigr| \biggr)^{\frac{1}{q}} \nonumber \\
&=  2^{\frac{1}{r}} \biggl( \sum_{y^n} \sum_{i=1}^n \rho(y_{0,i},y_i) \bigl|\mathbf{P}_{Y^n}(y^n) - \psi^n(y^n)\bigr| \biggr)^{\frac{1}{q}} \nonumber \\
&\leq 2^{\frac{1}{r}} \bigl( n \rho_{\max} \|\mathbf{P}_{Y^n} - \psi^n\|_{TV} \bigr)^{\frac{1}{q}} \nonumber \\
&\leq 2^{\frac{1}{r}} \bigl(n \rho_{\max} \exp{\{-\delta n\}} \bigr)^{\frac{1}{q}}, \nonumber \text{ by (\ref{aeq12})}.
\end{align}
Hence, we have
\begin{align}
\hat{T}_n(\mathbf{P}_{Y^n},\psi^n) \leq 2^{\frac{q}{r}} \rho_{\max} \exp{\{-\delta n\}}. \label{aeq13}
\end{align}
Recall that $\rho(x,y)=d(x,y)^p$ for some $p>0$. If $p\ge 1$, then   $\|V^n\|_p \coloneqq \big( E\big[\, \sum_{i=1}^n|V_i|^p\, \big] \big)^{1/p}$ is a norm on $\R^n$-valued random vectors whose components have finite $p$th moments, and if $1<p<0$, we still have
$\|U^n+V^n\|_p\le \|U^n\|_p+\|V^n\|_p$. Thus we can upper bound the distortion $
\cE[\rho_n(X^{n},\hat{Y}^{n})]$ of the code in Fig.~\ref{fig4} as follows:
\begin{align*}
&\biggl( \cE\biggl[ \frac{1}{n} \sum_{i=1}^n \rho(X_i,\hat{Y}_i)   \biggr] \biggr)^{1/q}
   = \biggl( \cE\biggl[ \frac{1}{n} \sum_{i=1}^n d(X_i,\hat{Y}_i)^p   \biggr] \biggr)^{1/q}
  \\*
&\le  \biggl( \cE\biggl[ \frac{1}{n} \sum_{i=1}^n d(X_i,Y_i)^p   \biggr]
\biggr)^{1/q} + \biggl( \cE\biggl[ \frac{1}{n} \sum_{i=1}^n d(Y_i,\hat{Y}_i)^p   \biggr] \biggr)^{1/q}
  \\
  &=   \Bigl(\cE[\rho_n(X^{n},Y^{n})]\Bigr)^{1/q} +
  \hat{T}_n(P_{Y^n},\psi^n)^{1/q},
\end{align*}
Hence, by (\ref{aeq11}) and (\ref{aeq13}) we obtain
\begin{align}
\limsup_{n\to \infty}  \cE[\rho_n(X^{n},\hat{Y}^{n})] \le D + \delta, \nonumber
\end{align}
which completes the proof.

\subsection{Achievability for Continuous Alphabets}
\label{sec achievcont}

In this section, we let $\sX = \sY = \R$, $\rho(x,y)=(x-y)^2$, and assume that $\mu$ and $\psi$ have finite second moments. We make use of the discrete case to prove the achievability for the continuous case.

Assume that $(R,R_c)$ is in the interior of $\L(D)$. Then there exists $P_{X,Y,U} \in \M(D)$ such that $R>I(X;U)$ and $R+R_c>I(Y;U)$. Let $q_k$ denote the uniform quantizer on the interval $[-k,k]$ having $2^k$ levels, the collection of which is denoted by $L_k$. Extend $q_k$ to the entire real line by using the nearest neighborhood encoding rule. Define $X(k) \coloneqq q_k(X)$ and $Y(k) \coloneqq q_k(Y)$. Let $\mu_k$ and $\psi_k$ denote the distributions of $X(k)$ and $Y(k)$, respectively. It is clear that
\begin{align}
\cE[(X-X(k))^2] \rightarrow 0, \text{  and  }  \cE[(Y-Y(k))^2] \rightarrow 0 \text{  as  } k\rightarrow\infty. \label{neq12}
\end{align}
Moreover, by \cite[Theorem 6.9]{Vil09} it follows that $\hat{T}(\mu_k,\mu) \rightarrow 0$ and $\hat{T}(\psi_k,\psi) \rightarrow 0$
as $k\rightarrow\infty$ since $\mu_k \rightarrow \mu$, $\psi_k \rightarrow \psi$ weakly \cite{Bil99}, and $\cE[X(k)^2]\rightarrow\cE[X^2]$, $\cE[Y(k)^2]\rightarrow\cE[Y^2]$. For each $k$ define $D_k \coloneqq \cE[(X(k)-Y(k))^2]$. Then by (\ref{neq12})
\begin{align}
\lim_{k\rightarrow\infty}  D_k = \cE[(X-Y)^2] \leq D. \nonumber
\end{align}

For any $k$, let $\M_k(D_k)$ be the set of distributions obtained by replacing $\mu$, $\psi$, and $\sX=\sY$ with $\mu_k$, $\psi_k$, and $\sX_k=\sY_k=L_k$, respectively, in (\ref{neq11}). Note that $X(k)-U-Y(k)$ and
\begin{align}
I(X(k);U) \leq I(X;U) \text{  and  } I(Y(k);U)\leq I(Y;U) \label{neq13}
\end{align}
by data processing inequality which implies $R>I(X(k);U)$ and $R+R_c>I(Y(k);U)$. Hence, $P_{X(k),Y(k),U} \in \M_k(D_k)$. Then, using the achievability result for discrete alphabets, for any $k$, one can find a sequence of $(n,R,R_c)^k$ randomized source codes for common source and reproduction alphabet $L_k$, source distribution $\mu_{k}$, and desired output distribution $\psi_{k}$ such that the upper limit of the distortions of these codes is upper bounded by $D_k$.

For each $k$ and $n$, consider the randomized source codes defined in Fig.~\ref{fig5}.
\begin{figure}[h]
\centering
\tikzstyle{int}=[draw, fill=white!20, minimum size=3em]
\tikzstyle{init} = [pin edge={to-,thin,black}]
\scalebox{0.6}{
\begin{tikzpicture}[node distance=3cm,auto,>=latex']
    \node [int] (a) {$T_{\mu_{k}^n|\mu^n}$};
    \node (b) [left of=a, coordinate] {a};
    \node [int] (c) [right of=a, node distance=4.2cm] {$(n,R,R_c)^k$};
    \node [int] (d) [right of=c,node distance=4.2cm] {$T_{\psi^n|\psi_{k}^n}$};
    \node [coordinate] (end) [right of=d]{};
    \path[->] (b) edge node {$\hat{X}^n(k)\sim\mu^n$} (a);
    \path[-] (a) edge node {$X^n(k)\sim\mu_{k}^n$} (c);
    \path[-] (c) edge node {$Y^n(k)\sim\psi_{k}^n$} (d);
    \draw[->] (d) edge node {$\hat{Y}^n(k)\sim\psi^{n}$} (end) ;
\end{tikzpicture}}
\caption{Randomized source code used in the achievability proof for continuous alphabets.}
\label{fig5}
\end{figure}
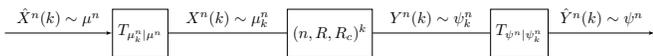
We note that the definition of the optimal transportation cost implies that $\hat{T}(\mu_{k}^n,\mu^n)\leq\hat{T}(\mu_k,\mu)$
and $\hat{T}(\psi_{k}^n,\psi^n)\leq\hat{T}(\psi_k,\psi)$. Hence, using the triangle inequality for the norm $\|V^n\|_2 \coloneqq\bigl( \sum_{i=1}^nE[ V_i^2] \bigl)^{1/2} $ on $\R^n$-valued random vectors having finite second moments, for all $k$, we have
\begin{align}
&\limsup_{n\rightarrow\infty} \cE\biggl[\bigl(\hat{X}^n(k)-\hat{Y}^n(k)\bigr)^2\biggr]^{1/2} \nonumber \\
&\phantom{xxxx}\leq \limsup_{n\rightarrow\infty} \biggl(\hat{T}(\mu_{k}^n,\mu^n)^{1/2}+\cE\biggl[\bigl(X^n(k)-Y^n(k)\bigr)^2\biggl]^{1/2}\nonumber \\
&\phantom{xxxxxxxxxxxxxxxxxxxxxxxxxxxx}+ \hat{T}(\psi_{k}^n,\psi^n)^{1/2}\biggr) \nonumber \\
&\phantom{xxxx}\leq \hat{T}(\mu_{k},\mu)^{1/2}+\hat{T}(\psi_{k},\psi)^{1/2}\nonumber\\
&\phantom{xxxxxxxxxxxxxxxx}+\limsup_{n\rightarrow\infty}\cE\biggl[\bigl(X^n(k)-Y^n(k)\bigr)^2\biggr]^{1/2} \nonumber \\
&\phantom{xxxx}\leq \hat{T}(\mu_{k},\mu)^{1/2}+\hat{T}(\psi_{k},\psi)^{1/2}+D_{k}^{1/2}. \nonumber
\end{align}
By choosing $k$ large enough we can make the last term arbitrarily close to $D$, which completes the proof.

\subsection{Cardinality Bound}
\label{sec card}

In this section, we show that for any discrete distribution $\Lambda_{X,Y,W}$ forming a Markov chain $X-W-Y$, there exists a discrete distribution $\Gamma_{X,Y,U}$ forming another Markov chain $X-U-Y$ such that
\begin{align}
|\sU| &\leq |\sX|+|\sY|+1, \nonumber \\
\Gamma_X &= \Lambda_X \nonumber \\
\Gamma_Y &= \Lambda_Y, \nonumber \\
\cE_{\Gamma}[\rho(X,Y)] &= \cE_{\Lambda}[\rho(X,Y)], \nonumber \\
I_{\Gamma}(X;U) &= I_{\Lambda}(X;W), \nonumber \\
I_{\Gamma}(Y;U) &= I_{\Lambda}(Y;W), \nonumber
\end{align}
where $I_{P}(X;U)$ denotes the mutual information computed with respect to the distribution $P$. Let $\P(\sX)\times\P(\sY)$ denote the product of probability simplices $\P(\sX)$ and $\P(\sY)$ representing the set of all distributions of independent random variables over $\sX\times\sY$. This set is compact and connected when viewed as a subset of $\R^{|\sX|+|\sY|}$. Without loss of generality $\sX=\{1,\ldots,|\sX|\}$ and $\sY=\{1,\ldots,|\sY|\}$. Since $H(X)$ is fixed in $I(X;W) = H(X) - H(X|W)$ (similarly $H(Y)$ is fixed in $I(Y;W) = H(Y) - H(Y|W)$), we define the following real valued continuous functions on $\P(\sX)\times\P(\sY)$:
\begin{align}
g_j(\nu)=\begin{cases}
\nu_x(j),   &\text{if } j=1,\ldots,|\sX|-1  \\
\nu_y(j),   &\text{if } j=|\sX|,\ldots,|\sX|+|\sY|-2  \\
\cE_{\nu}[\rho(X,Y)], &\text{if } j=|\sX|+|\sY|-1 \\
H(\nu_x), &\text{if } j=|\sX|+|\sY| \\
H(\nu_y), &\text{if } j=|\sX|+|\sY|+1,
\end{cases}
\nonumber
\end{align}
where $\nu=\nu_x\otimes\nu_y$  and $H(P)$ denotes the entropy of the distribution $P$. By so-called `support lemma' \cite[Appendix C]{GaKi11}, there exists a random variable $U \sim \Gamma_U$, taking values in $\sU$ with $ |\sU| \leq |\sX|+|\sY|+1$, and a conditional probability $\Gamma_{X|U}\Gamma_{Y|U}$ on $\sX\times\sY$ given $\sU$ such that for $j=1,\ldots,|\sX|+|\sY|+1$,
\begin{align}
&\sum_w g_j(\Lambda_{X|W=w}\Lambda_{Y|W=w}) \Lambda_W(w) \nonumber \\
&\phantom{xxxxxxxxxxxxxxxxx}= \sum_u g_j(\Gamma_{X|U=u}\Gamma_{Y|U=u}) \Gamma_U(u), \nonumber
\end{align}
which completes the proof.

\subsection{Converse}\label{sec conv}

We use the standard approach to prove the converse in Theorem~\ref{thm2}, i.e., that $\cl \cR(D) \subset \L(D)$ for any $D\geq0$. We note that this proof holds both for finite alphabets and continuous alphabets.

For each $R_c$, define the minimum coding rate $R$ at distortion level $D$ as
\[
\min \{R \in \cR(D,R_c)\} =: I_{R_c}(\mu\|\psi,D). \nonumber
\]
Using a time-sharing argument and the operational meaning of $I_{R_c}(\mu\|\psi,D)$, one can prove that
$I_{R_c}(\mu\|\psi,D)$ is convex in $D$, and therefore, continuous in $D$, $0< D <\infty$ (see the proof of Lemma~\ref{lemma1}).
Since $I_{R_c}(\mu\|\psi,D)$ is nonincreasing in $D$, we have $I_{R_c}(\mu\|\psi,0)\geq\lim_{D\rightarrow0}I_{R_c}(\mu\|\psi,D)$.
But by the definition of $\cR(0,R_c)$, we also have $\lim_{D\rightarrow0}I_{R_c}(\mu\|\psi,D) \in \cR(0,R_c)$, so that $I_{R_c}(\mu\|\psi,0)=\lim_{D\rightarrow0}I_{R_c}(\mu\|\psi,D)$. Hence, $I_{R_c}(\mu\|\psi,D)$ is also continuous at $D=0$.
Let us define $\cR^*(D) = \{(R,R_c) \in \R^2: R > I_{R_c}(\mu\|\psi,D)\}$ and let $(R,R_c) \in \cR^*(D)$. Since $I_{R_c}(\mu\|\psi,D)$ is continuous in $D$, there exists $\varepsilon>0$ such that $R>I_{R_c}(\mu\|\psi,D-\varepsilon)$. Hence, there exists, for all sufficiently large $n$, a $(n,R,R_c)$ randomized source code such that
\begin{align}
\cE[\rho_n(X^n,Y^n)] \leq D, \nonumber \\
Y^n \sim \psi^n. \nonumber
\end{align}
For each $n$, define the random variable $Q_n \sim \sUnif\{1,\ldots,n\}$ which is independent of $(X^n,Y^n,J,K)$, associated with the $n^{th}$ randomized source code. Since $J \in [2^{nR}]$,
\begin{align}
nR \geq H(J) \geq H(J|K) &\geq I(X^n;J|K)  \nonumber \\
&\overset{(a)}{=} I(X^n;J,K) \nonumber \\
&= \sum_{i=1}^n I(X_i;J,K|X^{i-1}) \nonumber \\
&\overset{(b)}{=} \sum_{i=1}^n I(X_i;J,K,X^{i-1}) \nonumber \\
&\geq \sum_{i=1}^n I(X_i;J,K) \nonumber \\
&= nI(X_{Q_n};J,K|Q_n) \nonumber \\
&\overset{(c)}{=} nI(X_{Q_n};J,K,Q_n), \nonumber
\end{align}
where $(a)$ follows from the independence of $X^n$ and $K$, $(b)$ follows from i.i.d. nature of the source $X^n$ and $(c)$ follows from the independence of $X_{Q_n}$ and $Q_n$. Similarly, for the sum rate we have
\begin{align}
n(R+R_c) \geq H(J,K) &\geq I(Y^n;J,K)  \nonumber \\
&= \sum_{i=1}^n I(Y_i;J,K|Y^{i-1}) \nonumber \\
&\overset{(a)}{=} \sum_{i=1}^n I(Y_i;J,K,Y^{i-1}) \nonumber \\
&\geq \sum_{i=1}^n I(Y_i;J,K) \nonumber \\
&= nI(Y_{Q_n};J,K|Q_n) \nonumber \\
&\overset{(b)}{=} nI(Y_{Q_n};J,K,Q_n), \nonumber
\end{align}
where $(a)$ follows from i.i.d. nature of the output $Y^n$ and $(b)$ follows from the independence of $Y_{Q_n}$ and $Q_n$. Notice that $X_{Q_n} \sim \mu$, $Y_{Q_n} \sim \psi$, and $X_{Q_n} - (J,K,Q_n) - Y_{Q_n}$. We also have
\begin{align}
\cE[\rho(X_{Q_n},Y_{Q_n})] &= \cE\biggl[\cE\bigl[\rho(X_{Q_n},Y_{Q_n})|Q_n\bigr]\biggr] \nonumber \\
&= \frac{1}{n} \sum_{i=1}^n \cE\bigl[\rho(X_{Q_n},Y_{Q_n})|Q_n=i\bigr] \nonumber \\
&= \frac{1}{n} \sum_{i=1}^n \cE\bigl[\rho(X_i,Y_i)\bigr] \nonumber \\
&= \cE\bigl[\rho_n(X^n,Y^n)\bigr] \leq D. \nonumber
\end{align}
Define $U = (J,K,Q_n)$ and denote by $P_{X,Y,U}$ the distribution of $(X_{Q_n},Y_{Q_n},U)$. Hence, $P_{X,Y,U} \in \M(D)$ which implies that $(R,R_c) \in \L(D)$. Hence, $\cR^*(D) \subset \L(D)$. But, since $\L(D)$ is closed in $\R^2$, we also have $\cl \cR^*(D) = \cl \cR(D) \subset \L(D)$.

\section{Conclusion}
\label{conc}

Generalizing the practically motivated distribution preserving quantization
problem, we have derived the rate distortion region for randomized source coding
of a stationary and memoryless source, where the output of the code is
restricted to be also stationary and memoryless with some specified
distribution. For a given distortion level, the rate region consists of coding
and common randomness rate pairs, where the common randomness is independent of
the source and shared between the encoder and the decoder. Unlike in classical
rate distortion theory, here shared independent randomness can decrease the
necessary coding rate communicated between the encoder and decoder.

\section*{Appendix}

\appsec

\subsection{Proof of Lemma~\ref{lemma1}}\label{sec6sub1}

Let $D_1$ and $D_2$ be two distinct positive real numbers and choose $\alpha \in (0,1)$. Fix any $\varepsilon > 0$. Let $\delta$ be a small positive number which will be specified later. By the definition of $I_0(\mu\|\psi,D)$ and by Theorem~\ref{thm2} there exist positive real numbers $R_1$ and $R_2$ such that
\begin{align}
R_i &\leq I_0(\mu\|\psi,D_i) + \delta, i=1,2, \nonumber
\end{align}
and such that for all sufficiently large $n$ there exist randomized $(n,R_1,0)$ and $(n,R_2,0)$ source codes having output distribution $\psi^n$ which satisfy
 \begin{align}
\cE\biggl[\rho_n\biggl(X^n,F^{(1)}\bigl(E^{(1)}(X^n)\bigr)\biggr)\biggr] \leq D_1 + \delta, i=1,2, \nonumber
\end{align}
where $(E^{(1)},F^{(1)})$ and $(E^{(2)},F^{(2)})$ are the encoder-decoder pairs for these codes. Let $\{k_M\}_{M\geq1}$ be a sequence of positive integers such that $\lim_{M\rightarrow\infty} \frac{k_M}{M} = \alpha$. Let $N$ be a positive integer which will be specified later. For the source block $X^{nN}$ define the following randomized source code:
\begin{align}
E &\coloneqq \bigl(\underbrace{E^{(1)},\ldots, E^{(1)}}_{\text{$k_N$-times}},\underbrace{E^{(2)},\ldots, E^{(2)}}_{\text{$N - k_N$-times}}\bigr), \nonumber \\
F &\coloneqq \bigl(\underbrace{F^{(1)},\ldots, F^{(1)}}_{\text{$k_N$-times}},\underbrace{F^{(2)},\ldots, F^{(2)}}_{\text{$N - k_N$-times}}\bigr) \nonumber.
\end{align}
Note that the output distribution for this randomized source code is $\psi^{nN}$, and its rate $R$ and distortion $D$ satisfy the following
\begin{align}
R &= \frac{1}{n N} \bigl( k_N n R_1 + ( N -  k_N ) n R_2\bigr) \nonumber \\
&\leq \frac{k_N}{N} I_0(\mu\|\psi,D_1) + \frac{N-k_N}{N} I_0(\mu\|\psi,D_2) + \delta, \nonumber\\
\intertext{and}
D &= \cE\bigl[\rho_{n N}(X^{n N},Y^{n N})\bigr] \leq \frac{k_N}{N} D_1 + \frac{N-k_N}{N} D_2 + \delta. \nonumber
\end{align}
Since $\lim_{M\rightarrow\infty} \frac{k_M}{M} = \alpha$, one can choose $N$ and $\delta$ such that $R$ is upper bounded by $\alpha I_0(\mu\|\psi,D_1) + (1-\alpha) I_0(\mu\|\psi,D_2) + \varepsilon$ and $D$ is upper bounded by $\alpha D_1 + (1-\alpha) D_2+\varepsilon$. By Definition~\ref{def1}, this yields
\begin{align}
&I_0\bigl(\mu\|\psi,\alpha D_1 + (1-\alpha) D_2\bigr) \nonumber \\
&\phantom{xxxxxxxx}\leq \alpha I_0(\mu\|\psi,D_1) + (1-\alpha) I_0(\mu \|\psi,D_2) + \varepsilon. \nonumber
\end{align}
Since $\varepsilon$ is arbitrary, this completes the proof.

\subsection{Proof of Corollary~\ref{cor1}}\label{sec6sub2}

Assume that $(R,R_c)$ is in the interior of $\S(D)$. Then there exists $P_{X,Y,U} \in \H(D)$ such that $R > I(X;U)$ and $R+R_c > I(X,Y;U)$. Let $\pi = P_{X,Y}$. By Theorem~\ref{thm5} there exists a sequence of $(n,R,R_c)$ randomized source codes such that
\begin{align}
\lim_{n\rightarrow\infty} \|P_{X^n,Y^n}-\pi^n\| = 0, \label{aeq16}
\end{align}
where $(X^n,Y^n)$ denotes the input-output of the $n^{th}$ code.
Since $\rho_n$ is bounded, we have
\begin{align}
&\limsup_{n\rightarrow\infty} \bigl| \cE[\rho_n(X^n,Y^n)] - D \bigr| \nonumber \\
&\phantom{xxxx}= \limsup_{n\rightarrow\infty} \bigl| \cE[\rho_n(X^n,Y^n)] - \cE_{\pi^n}[\rho_n(X^n,Y^n)] \bigr|\nonumber \\
&\phantom{xxxx}\leq \limsup_{n\rightarrow\infty}   \|P_{X^n,Y^n} - \pi^n\|_{TV} \rho_{\max} = 0 , \label{eq44}
\end{align}
where $\cE_{\pi^n}$ denotes the expectation with respect to $\pi^n$.
Let $T_{\hat{Y}^n|Y^n}$ be the optimal coupling (i.e., conditional probability) between $P_{Y^n}$ and $\psi^n$ with the transportation cost $\hat{T}(P_{Y^n},\psi^n)$ with cost function $\rho_n$. By \cite[Theorem 6.15]{Vil09} and (\ref{aeq16}) one can prove that $\limsup_{n\rightarrow\infty}\hat{T}(P_{Y^n},\psi^n)=0$ as in (\ref{aeq13}).

For each $n$, let us define the following encoder-decoder pair (see Fig.~\ref{fig3})
\begin{align}
\tilde{E}^n_{J|X^n,K} &\coloneqq E^{n}_{J|X^n,K} \label{eq46} \\
\tilde{F}^n_{\hat{Y}^n|J,K} &\coloneqq T_{\hat{Y}^n|Y^n}\circ F^{n}_{Y^n|J,K} \label{eq47},
\end{align}
where $(E^n,F^n)$ is the encoder-decoder pair of the $n^{th}$ code.
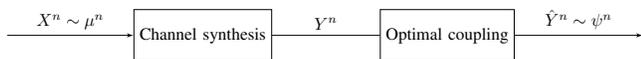
\begin{figure}[h]
\centering
\tikzstyle{int}=[draw, fill=white!20, minimum size=3em]
\tikzstyle{init} = [pin edge={to-,thin,black}]
\scalebox{0.65}{
\begin{tikzpicture}[node distance=5cm,auto,>=latex']
    \node [int] (a) {Channel synthesis};
    \node (b) [left of=a,node distance=4cm, coordinate] {a};
    \node [int] (c) [right of=a] {Optimal coupling};
    \node [coordinate] (end) [right of=c, node distance=4cm]{};
    \path[->] (b) edge node {$X^n\sim\mu^n$} (a);
    \path[-] (a) edge node {$Y^n$} (c);
    \draw[->] (c) edge node {$\hat{Y}^n\sim\psi^{n}$} (end) ;
\end{tikzpicture}}
\caption{Sub-optimal randomized source code achieving the rate
    region of  Corollary~\ref{cor1}.}
\label{fig3}
\end{figure}
Note that the randomized source code defined in (\ref{eq46}) and (\ref{eq47}) has rates $(R,R_c)$ and output distribution $\psi^n$. Furthermore, using the triangle inequality as in Section~\ref{sec achiev} one can prove that
\begin{align}
\limsup_{n\to \infty}  \cE[\rho_n(X^{n},\hat{Y}^{n})] \le D \nonumber
\end{align}
using (\ref{eq44}) and the fact that $\limsup_{n\rightarrow\infty}\hat{T}(P_{Y^n},\psi^n)=0$. This completes the proof.

\subsection{Proof of Theorem~\ref{thm3}}\label{sec6sub2a}

Since $\cR_e(D,R_c) \supset \cR_e(D,0)$ for all $R_c$, it is enough to prove that
\begin{align}
\cR_e(D,0) &\supset \L(D,\infty), \nonumber \\
\cR_e(D,R_c) &\subset \L(D,\infty). \nonumber
\end{align}
Recall that
\begin{align}
\L(D,\infty) = \{R \in \R: \exists P_{X,Y} \in \G(D) \text{  s.t.  } R\geq I(X;Y)\}. \nonumber
\end{align}
Let us assume that $R \in \L(D,\infty)$. Then, there exists $P_{X,Y} =: \pi \in \G(D)$ such that $R \geq I(X;Y)$. Fix any $\varepsilon > 0$. By Theorem~\ref{thm6} there exists a sequence of $(n,R,\infty)$ randomized source codes such that
\begin{align}
\lim_{n \rightarrow \infty} \|p_{X^n,Y^n} - \pi\|_{TV} &= 0 \text{ in probability}, \label{eq54}
\intertext{which implies}
\lim_{n \rightarrow \infty} \|p_{Y^n} - \psi\|_{TV} &= 0 \text{ in probability}. \nonumber
\end{align}
Hence, this sequence of codes satisfies the second constraint in Definition~\ref{def3}. To show that the codes satisfy the distortion constraint, we use the same steps in \cite[Section VII-D]{Cuf10}. We have
\begin{align}
\rho_n(X^n,Y^n) &= \frac{1}{n} \sum_{i=1}^n \rho(X_i,Y_i) \nonumber \\
&= \frac{1}{n} \sum_{i=1}^n \sum_{x,y} 1_{\{X_i=x, Y_i=y\}} \rho(x,y) \nonumber \\
&= \sum_{x,y} \rho(x,y) \frac{1}{n} \sum_{i=1}^n 1_{\{X_i=x, Y_i=y\}} \nonumber \\
&= \cE_{p_{X^n,Y^n}}[\rho(X,Y)], \nonumber
\end{align}
where $1_B$ denotes the indicator of event $B$ and $\cE_{p_{X^n,Y^n}}$ denotes the expectation with respect to the empirical distribution $p_{X^n,Y^n}$ of $(X^n,Y^n)$. For any $\varepsilon_1 > 0$, by (\ref{eq54}) we have
\begin{align}
\sPr\biggl\{\|p_{X^n,Y^n} - \pi\|_{TV} > \varepsilon_1\biggr\} < \varepsilon_1,\nonumber
\end{align}
for all sufficiently large $n$. Define the event $B_{\varepsilon_1} \coloneqq \bigl\{\|p_{X^n,Y^n}-\pi\|_{TV} \leq \varepsilon_1\bigr\}$. Then, for all sufficiently large $n$, we obtain
\begin{align}
&\cE[\rho_n(X^n,Y^n)] \nonumber \\
&\phantom{x}= \cE \biggl[ \cE_{p_{X^n,Y^n}}\bigl[\rho(X,Y)\bigr] \biggr]\nonumber \\
&\phantom{x}= \cE \biggl[ \cE_{p_{X^n,Y^n}}\bigl[\rho(X,Y)\bigr] 1_{B_{\varepsilon_1}} \biggr]+\cE \biggl[ \cE_{p_{X^n,Y^n}}\bigl[\rho(X,Y)\bigr] 1_{B_{\varepsilon_1}^c } \biggr] \nonumber \\
&\phantom{x}\leq \cE \biggl[ \cE_{p_{X^n,Y^n}}\bigl[\rho(X,Y)\bigr] 1_{B_{\varepsilon_1}} \biggr]+\rho_{\max} \varepsilon_1 \nonumber \\
&\phantom{x}\leq \cE_{\pi}\bigl[\rho(X,Y)\bigr] + 2\varepsilon_1\rho_{\max} \nonumber \\
&\phantom{x}\leq D + 2\varepsilon_1\rho_{\max}. \nonumber
\end{align}
By choosing $\varepsilon_1$ such that $2\varepsilon_1\rho_{\max} < \varepsilon$, we obtain $\cR_e(D,0) \supset \L(D,\infty)$.

To prove $\cR_e(D,R_c) \subset \L(D,\infty)$, we use the same arguments as in \cite[Section VII-B]{Cuf10}. Let us choose $R \in \cR_e(D,R_c)$ with the corresponding sequence of $(n,R,R_c)$ randomized source codes satisfying constraints in Definition~\ref{def3}. For each $n$, define the random variable $Q_n \sim \sUnif\{1,\ldots,n\}$ which is independent of the input-output $(X^n,Y^n)$ of the code $(n,R,R_c)$. Then, we have
\begin{align}
nR &\geq H(J) \nonumber \\
&\geq I(X^n;Y^n) \nonumber \\
&= \sum_{i=1}^n I(X_i;Y^n|X^{i-1}) \nonumber \\
&= \sum_{i=1}^n I(X_i;Y^n,X^{i-1}) \nonumber \\
&\geq \sum_{i=1}^n I(X_i;Y_i) \nonumber \\
&= n I(X_{Q_n};Y_{Q_n}|Q_n) \nonumber \\
&\overset{(a)}{=} n I(X_{Q_n};Y_{Q_n},Q_n) \nonumber \\
&\geq n I(X_{Q_n};Y_{Q_n}), \label{eq56}
\end{align}
where $(a)$ follows from the independence of $X_{Q_n}$ and $Q_n$. We also have
\begin{align}
\cE[\rho(X_{Q_n},Y_{Q_n})] &= \cE\biggl[\cE\bigl[\rho(X_{Q_n},Y_{Q_n})|Q_n\bigr]\biggr] \nonumber \\
&= \frac{1}{n} \sum_{i=1}^n \cE\bigl[\rho(X_{Q_n},Y_{Q_n})|Q_n=i\bigr] \nonumber \\
&= \frac{1}{n} \sum_{i=1}^n \cE\bigl[\rho(X_i,Y_i)\bigr] \nonumber \\
&= \cE\bigl[\rho_n(X^n,Y^n)\bigr]. \label{eq57}
\end{align}
One can prove
$P_{Y_{Q_n}} \rightarrow \psi$  in total variation  (see, e.g., \cite[Section VII-B-3]{Cuf10}). Since the set of probability distributions over $\sX\times\sY$ is compact with respect to the total variation distance, we can find a subsequence $\{(X_{Q_{n_k}},Y_{Q_{n_k}})\}$ of $\{(X_{Q_{n}},Y_{Q_{n}})\}$ such that
\begin{align}
P_{X_{Q_{n_k}},Y_{Q_{n_k}}} \rightarrow P_{\hat{X},\hat{Y}} \nonumber
\end{align}
in total variation for some $P_{\hat{X},\hat{Y}}$. But, since $P_{X_{Q_{n_k}}}=\mu$ for all $k$ and $P_{Y_{Q_{n}}} \rightarrow \psi$ in total variation, we must have $P_{\hat{X}}=\mu$ and $P_{\hat{Y}}=\psi$. Now, taking the limit of (\ref{eq56}) and (\ref{eq57}) through this subsequence, we obtain
\begin{align}
R &\geq \lim_{k\rightarrow\infty} I(X_{Q_{n_k}};Y_{Q_{n_k}}) = I(\hat{X};\hat{Y}) \nonumber \\
\intertext{and}
\cE[\rho(\hat{X},\hat{Y})] &= \lim_{k\rightarrow\infty} \cE[\rho(X_{Q_{n_k}},Y_{Q_{n_k}})] \nonumber\\
&= \lim_{k\rightarrow\infty} \cE[\rho_{n_k}(X^{n_k},Y^{n_k})] \leq D. \nonumber
\end{align}
Hence, $R \in \L(D,\infty)$ which completes the proof.

\subsection{Proof of Theorem~\ref{thm4}} \label{sec6sub3}

\paragraph{Achievability}

Assume $(R,R_c)$ is in the interior of $\cl R_{dd}(D)$. Then there exists
$P_{X,Y} =: \pi \in \G(D)$ such that $R > I(X;Y)$ and $R+R_c > H(Y)$. By
\cite[Theorem 1]{BeDeHaShWi13} or \cite[Section III-E]{Cuf13}, there exists a
sequence of $(n,R,R_c)$ randomized source codes with deterministic decoder such
that
\begin{align}
\|P_{X^n,Y^n} - \pi^n\|_{TV} \rightarrow 0. \nonumber
\end{align}
Hence, $\|P_{Y^n} - \psi^n\|_{TV} \rightarrow 0$ and
\begin{align}
\lim_{n\rightarrow\infty} \cE[\rho_n(X^n,Y^n)] = \lim_{n\rightarrow\infty} \cE_{\pi^n}[\rho_n(X^n,Y^n)] \leq D \nonumber
\end{align}
completing the proof.

\paragraph{Converse}

Let $(R,R_c) \in \cl \mathcal{R}_{dd}(D)$. Using a similar argument as in Appendix~\ref{sec6sub2a}, one can show that
\begin{align}
nR &\geq n I(X_{Q_n};Y_{Q_n}), \label{eq70}
\intertext{and}
\cE[\rho(X_{Q_n},Y_{Q_n})] &= \cE[\rho_n(X^n,Y^n)], \label{eq71}
\end{align}
where $Q_n \sim \sUnif\{1,\ldots,n\}$ is independent of input-output $(X^n,Y^n)$ of the corresponding randomized source code, and $P_{Y_{Q_n}} \rightarrow \psi$ in total variation. Also, there is a subsequence $\{(X_{Q_{n_k}},Y_{Q_{n_k}})\}$ such that
$P_{X_{Q_{n_k}},Y_{Q_{n_k}}} \rightarrow P_{\hat{X},\hat{Y}}$ in total variation
for some $P_{\hat{X},\hat{Y}}$ with $P_{\hat{X}}=\mu$ and $P_{\hat{Y}}=\psi$. By taking the limit of (\ref{eq70}) and (\ref{eq71}) through this subsequence we obtain
\begin{align}
R &\geq I(\hat{X};\hat{Y}), \label{eq60} \\
\cE[\rho(\hat{X},\hat{Y})] &\leq D. \label{eq72}
\end{align}
Hence, the first inequality in (\ref{eq21}) is satisfied.  To show the second inequality, let $a_n := \|P_{Y^n} - \psi^n\|_{TV}$. By \cite[Theorem 17.3.3]{CoTh06}, we have
\begin{align}
|H(Y^n) - H(\psi^n)| \leq a_n \log\biggl(\frac{|\sY|^n}{a_n}\biggr), \nonumber
\end{align}
where $H(\psi^n)= n H(\psi)$. Since the decoder is a deterministic function of
$J$ and $K$, we have
\begin{align}
n H(\psi) - a_n \bigl( n \log |\sY|- \log a_n  \bigr) &\leq H(Y^n) \leq n(R+R_c). \nonumber
\end{align}
Since $a_n \rightarrow 0$ as $n\rightarrow\infty$, this yields $R+R_c \geq H(\psi) = H(Y)$.

\section*{Acknowledgement} The authors would like to thank two anonymous
reviewers for many constructive comments.

\bibliographystyle{IEEEtran}

\end{document}